\documentclass[11pt]{article}
\usepackage{geometry}
\usepackage{amsmath, amssymb, amsthm}
\usepackage{graphicx}
\usepackage{bm}
\usepackage[linesnumbered,ruled]{algorithm2e}
\newtheorem{theorem}{Theorem}
\newtheorem{lemma}{Lemma}

\SetKwRepeat{Do}{do}{while}
\title{Strongly Stable Matchings under Matroid Constraints}
\author{Naoyuki Kamiyama%
\thanks{Institute of Mathematics for Industry, Kyushu University, Fukuoka, Japan.}
\thanks{{\ttfamily kamiyama@imi.kyushu-u.ac.jp}}
\thanks{This work was supported by JSPS KAKENHI Grant Number JP20H05795.}
}
\date{}

\begin{document}

\maketitle

\begin{abstract}
We consider a many-to-one variant of 
the stable matching problem. 
More concretely, we consider the variant of 
the stable matching problem where one side has a matroid constraint. 
Furthermore, we consider the situation where 
the preference of each agent may contain ties. 
In this setting, we consider the problem of checking
the existence of a strongly stable matching, and 
finding a strongly stable matching if a strongly stable matching
exists. 
We propose a polynomial-time algorithm for this problem. 
\end{abstract}

\section{Introduction}

The topic of this paper is a matching problem between 
two groups. Especially, we focus on a model where
each agent has a preference over agents in the other group. 
The stable matching problem, which was proposed by 
Gale and Shapley~\cite{GS62}, is 
one of the most fundamental models. 
In the setting considered in \cite{GS62}, the 
preference of each agent is strict. 
In contrast, we consider the situation where 
the preference of each agent 
may contain ties. 
That is, each agent may be indifferent between 
potential partners. 
It is known that if 
the preferences may contain ties, then the situation 
dramatically changes (see, e.g., \cite{IM08} and 
\cite[Chapter~3]{M13}). 
Under preferences with ties, there exist three 
stability concepts. 
The first concept is called weak stability. 
This property guarantees that 
there does not exist an unmatched pair $x,y$ such that 
$x$ (resp.\ $y$) prefers $y$ (resp.\ $x$) to the current partner.
Irving~\cite{I94} proved 
that there always exists a weakly stable matching, and 
we can find a weakly stable matching
in polynomial time by slightly modifying the algorithm 
of Gale and Shapley~\cite{GS62}. 
The second concept is called super-stability.
This property guarantees
that there does not exist an unmatched pair $x,y$ 
such that $x$ (resp.\ $y$) prefers $y$ (resp.\ $x$) to the current partner, or is 
indifferent between $y$ (resp.\ $x$) and the current partner.
The last concept is called strong stability,
which 
is the main topic of this paper. 
This property guarantees that 
there does not exist an unmatched pair $x,y$ such that
(i) $x$ prefers $y$ to 
the current partner, and (ii) $y$ 
prefers $x$ to the current partner, or is indifferent between 
$x$ and the current partner. 
It is known that 
a super-stable matching and a strongly stable matching 
may not exist (see \cite{I94}). 

Furthermore, we consider a matroid constraint.   
Matroids can express not only capacity constraints 
but also more complex constraints including hierarchical 
capacity constraints.
Thus, matroid constraints are important from not only the theoretical 
viewpoint but also the practical viewpoint. 
Matroid approaches to 
the stable matching problem 
have been extensively 
studied (see, e.g., \cite{K15,K19,K20,K22,F03,FT07,FK16,Y17,MY15,IwataY20,KTY18}).
Under matroid constraints, 
Kamiyama~\cite{K22} proposed a polynomial-time 
algorithm for the problem of checking
the existence of a super-stable matching, and 
finding a super-stable matching if 
a super-stable matching exists. 
(Precisely speaking, in \cite{K22}, 
the many-to-many variant was considered.) 

{\bf Our contribution.}
In this paper, we consider the problem of checking
the existence of a strongly stable matching, and 
finding a strongly stable matching if a strongly stable matching
exists in the many-to-one variant 
of the stable matching problem with 
matroid constraints. 
We propose a polynomial-time algorithm for this problem. 
Although special cases of this problem were considered in \cite{K15,K19}, 
the polynomial-time solvability of this problem
has been open.
In this paper, we affirmatively settle this question. 

Our algorithm is based on 
the algorithm proposed by 
Olaosebikan and Manlove~\cite{OM20} for 
the student-project allocation problem with ties.
The extension needs non-trivial observations 
about matroid constraints. 
Furthermore, when we express the 
student-project allocation problem with ties 
in our matroid model, the definition of 
a blocking pair is slightly different from that 
in \cite{OM20}. 
See Section~\ref{appendix:spa} for details. 

{\bf Related work.} 
In the one-to-one setting, 
Irving~\cite{I94} 
proposed polynomial-time algorithms for 
the super-stable matching problem and 
the strongly stable matching problem 
(see also \cite{M99}). 
Kunysz, Paluch, and Ghosal~\cite{KunyszPG16}
considered 
characterization of the set of all strongly stable matchings. 
Kunysz~\cite{Kunysz18} considered 
the weighted version of the strongly stable matching
problem. 
In the many-to-one setting, 
Irving, Manlove, and Scott~\cite{IMS00} proposed 
a polynomial-time algorithm for 
the super-stable matching problem.
Irving, Manlove, and Scott~\cite{IMS03} and 
Kavitha, Mehlhorn, Michail, and Paluch~\cite{KMMP07}
proposed polynomial-time algorithms for the 
strongly stable matching
problem. 
In the many-to-many setting, 
Scott~\cite{S05} considered 
the super-stable matching problem, 
and the papers~\cite{Malhotra04,ChenG10,Kunysz19} considered 
the strongly stable matching problem.
Furthermore, Olaosebikan and Manlove~\cite{OlaosebikanM22,OM20}
considered super-stability and strong stability in the 
student-project allocation problem. 

For the situation where a master list is given, 
Irving, Manlove, and Scott~\cite{IMS08} gave 
simple polynomial-time algorithms for the 
super-stable matching problem 
and the strongly stable matching problem. 
O'Malley~\cite{O07} gave 
polynomial-time algorithms for 
the 
super-stable matching problem 
and the strongly stable matching problem
in the many-to-one setting.
Under matroid constraints, 
Kamiyama~\cite{K15,K19} gave 
polynomial-time algorithms for 
the super-stable matching problem 
and the strongly stable matching problem
in the many-to-one and many-to-many settings. 

\section{Preliminaries} 

We denote by $\mathbb{Z}$ and $\mathbb{Z}_+$ 
the sets of integers and non-negative integers, respectively.  
For each positive integer $z$, we define 
$[z] := \{1,2,\ldots,z\}$. 
Define $[0] := \emptyset$. 
For each finite set $X$ and each element $u$, we define 
$X + u := X \cup \{u\}$ and $X - u := X \setminus \{u\}$.

An ordered pair ${\bf M} = (U, \mathcal{I})$ of a finite set $U$ and 
a non-empty family $\mathcal{I}$ of subsets of $U$ 
is called a 
{\it matroid} if 
the following conditions are satisfied 
for every pair of 
subsets $I,J \subseteq U$. 
\begin{itemize}
\item[\bf (I1)]
If $I \subseteq J$ and $J \in \mathcal{I}$, then 
$I \in \mathcal{I}$. 
\item[\bf (I2)]
If $I, J \in \mathcal{I}$ and $|I| < |J|$, then 
there exists an element $u \in J \setminus I$ such that 
$I + u \in \mathcal{I}$. 
\end{itemize} 
If ${\bf M} = (U, \mathcal{I})$ is a matroid, then 
we say that \emph{${\bf M}$ is defined on a ground set $U$}. 
In addition, an element in 
$\mathcal{I}$ is called an {\em independent set of ${\bf M}$}. 

Let $U$ be a finite set. 
A function $\rho \colon 2^U \to \mathbb{Z}$ 
is said to be \emph{submodular} if 
$\rho(X) + \rho(Y) \ge \rho(X \cup Y) + \rho(X \cap Y)$
for every pair of subsets $X,Y \subseteq U$. 
For each function $\rho \colon 2^U \to \mathbb{Z}$, 
a subset $X \subseteq U$ is called a \emph{minimizer of $\rho$} if 
$X$ minimizes $\rho(X)$ 
among all the subsets of $U$. 
For each 
function $\rho \colon 2^U \to \mathbb{Z}$, 
a minimizer $X$ of $\rho$ is said to be 
\emph{minimal} if 
there does not exist a minimizer $Y$ of $\rho$ such that 
$Y \subsetneq X$. 
For each submodular function 
$\rho \colon 2^U \to \mathbb{Z}$, it is known that if 
$\rho(X)$ can be evaluated 
in time bounded by a 
polynomial in $|U|$ for every subset $X \subseteq U$, then 
we can find a minimizer $X$ of $\rho$
in time bounded by a 
polynomial in $|U|$ (see, e.g., \cite{S00,IFF01}). 
Furthermore, under the same condition, it is known 
that a minimal minimizer of $\rho$ is uniquely determined, and 
we can find the minimal minimizer 
of $\rho$ in time bounded by a polynomial in $|U|$
(see, e.g., \cite[Note~10.12]{M03}). 

In a finite simple
directed graph, an arc from a vertex $u$ to a vertex $v$ 
is denoted by $uv$. 
A directed cycle $P$ in a finite simple directed graph is said to be 
\emph{simple} if $P$ passes through each vertex at most once. 
We may not distinguish between 
a simple directed cycle $P$ and the set of arcs 
that $P$ passes through. 
For each simple directed cycle $P$ 
in a finite simple directed graph, an arc $a$ that $P$ does not pass through 
is called a \emph{shortcut arc for $P$} if 
$P + a$ contains a simple directed cycle that is different 
from $P$. 
It is not difficult to see that if 
there exists a directed cycle in a finite simple directed graph, then 
there exists a simple directed cycle for which 
there does not exist a shortcut arc. 

\subsection{Problem formulation} 

Throughout this paper, let 
$G$ be a finite simple undirected bipartite graph. 
We denote by $V$ and $E$ the vertex set of $G$ and 
the edge set of $G$, respectively. 
We assume that $V$ is partitioned into $D$ and $H$, and 
every edge in $E$ connects a vertex in $D$ and a vertex in $H$. 
We call a vertex in $D$ (resp.\ $H$) 
a {\it doctor} (resp.\ {\it hospital}). 
For each doctor $d \in D$ and each hospital 
$h \in H$, if there exists the edge in $E$ 
between $d$ and $h$, then we denote by $(d,h)$ this edge. 
For each subset $F \subseteq E$ and each 
subset $X \subseteq D$, 
we denote by $F(X)$ the set of edges $(d, h) \in F$ 
such that $d \in X$.
For each subset $F \subseteq E$ and 
each doctor $d \in D$, 
we write $F(d)$ instead of $F(\{d\})$.  

For each doctor $d \in D$, we are given a transitive binary 
relation $\succsim_d$ on $E(d) \cup \{\emptyset\}$ such that, for every pair of 
elements $e,f \in E(d) \cup \{\emptyset\}$, at least one of $e \succsim_d f$ and 
$f \succsim_d e$ holds. 
For each doctor $d \in D$ and each pair of 
elements $e,f \in E(d) \cup \{\emptyset\}$, 
if $e \succsim_d f$ and 
$f \not \succsim_d e$ 
(resp.\ $e \succsim_d f$ and 
$f \succsim_d e$), then 
we write $e \succ_d f$
(resp.\ $e \sim_d f$). 
In this paper, 
we assume that,
for every doctor $d \in D$ and every edge $e \in E(d)$, 
we have 
$e \succ_d \emptyset$. 

We are given a matroid ${\bf H} = (E, \mathcal{F})$
such that $\{e\} \in \mathcal{F}$ for every edge $e \in E$.
In addition, we are given 
a transitive binary 
relation $\succsim_H$ on $E$ such that, for every pair of 
edges $e,f \in E$, at least one of $e \succsim_H f$ and 
$f \succsim_H e$ 
holds. 
For each pair of 
edges $e,f \in E$, 
if $e \succsim_H f$ and 
$f \not \succsim_H e$ 
(resp.\ $e \succsim_H f$ and 
$f \succsim_H e$), then 
we write $e \succ_H f$
(resp.\ $e \sim_H f$). 
We assume that, 
for every subset $F \subseteq E$, we can 
determine whether $F \in \mathcal{F}$ in time bounded by a polynomial 
in the input size of $G$.  
(That is, we consider the independence oracle model.) 

A subset $\mu \subseteq E$ is called a {\it matching in $G$}
if the following conditions are satisfied. 
\begin{itemize}
\item[\bf (M1)]
$|\mu(d)| \le 1$ for every doctor $d \in D$. 
\item[\bf (M2)]
$\mu \in \mathcal{F}$. 
\end{itemize} 
For each matching $\mu$ in $G$ and each doctor 
$d \in D$ such that 
$\mu(d) \neq \emptyset$, 
we do not distinguish between $\mu(d)$ and the unique element in 
$\mu(d)$. 

Let $\mu$ be a matching in $G$. 
We say that 
an edge $e = (d,h) \in E \setminus \mu$ 
\emph{weakly} (resp.\ \emph{strongly}) \emph{blocks $\mu$ on $d$} if 
$e \succsim_d \mu(d)$
(resp.\ $e \succ_d \mu(d)$). 
In addition, 
we say that 
an edge $e \in E \setminus \mu$ 
\emph{weakly} (resp.\ \emph{strongly}) \emph{blocks $\mu$ on ${\bf H}$} if 
one of the following conditions is satisfied. 
\begin{itemize}
\item[\bf (H1)]
$\mu + e \in \mathcal{F}$. 
\item[\bf (H2)] 
$\mu + e \notin \mathcal{F}$, and 
there exists an edge $f \in \mu$ such that 
$e \succsim_H f$
(resp.\ $e \succ_H f$)
and 
$\mu + e - f \in \mathcal{F}$. 
\end{itemize}
Then an edge $e = (d,h) \in E \setminus \mu$ 
is called a \emph{blocking edge for $\mu$} if the following conditions are 
satisfied. 
\begin{itemize}
\item[\bf (B1)]
$e$ weakly blocks $\mu$ on
$d$ and ${\bf H}$. 
\item[\bf (B2)]
$e$ strongly blocks $\mu$ on
at least one of $d$ and ${\bf H}$. 
\end{itemize}
Finally, $\mu$ is called a \emph{strongly 
stable matching in $G$} if there does not exist a 
blocking edge $e \in E \setminus \mu$ for $\mu$. 

Here we give a concrete example of our abstract model. 
For each hospital $h \in H$, we are given a  
family $\mathcal{P}_h$ of subsets of $E(h)$
such that, 
for every pair of distinct elements $X,Y \in \mathcal{P}_h$, 
$X \cap Y = \emptyset$ or  
$X \subseteq Y$ or $Y \subseteq X$. 
That is, $\mathcal{P}_h$ is a laminar family. 
In addition, for each hospital $h \in H$, we are given 
a capacity function $c_h \colon \mathcal{P}_h \to \mathbb{Z}_+$. 
Define 
$\mathcal{F}$ as the set of subsets $F \subseteq E$ 
such that,
for every 
hospital $h \in H$ and every element $P \in \mathcal{P}_h$, 
we have $|F(h) \cap P| \le c_ h(P)$. 
It is well known that the ordered pair 
${\bf H} = (E, \mathcal{F})$ defined in this way is a matroid. 
For example, laminar constraints in the stable matching 
problem are considered in
\cite{H10}. 

\subsection{Basics of matroids} 

Throughout this subsection, 
let ${\bf M} = (U, \mathcal{I})$ be a matroid.

A subset $B \subseteq U$ is called a {\em base of ${\bf M}$} if 
$B$ is an inclusion-wise maximal independent set of ${\bf M}$. 
Notice that (I2) implies that 
all the bases of ${\bf M}$ have the same size. 
We can find a base of ${\bf M}$ 
by adding elements in $U$ one by one and checking whether 
the resulting subset is an independent set. 
A subset $I \subseteq U$ such that $I \notin \mathcal{I}$ is called a 
{\em dependent set of ${\bf M}$}.
Furthermore, a subset $C \subseteq U$ is called a {\em circuit of ${\bf M}$} if 
$C$ is an inclusion-wise minimal dependent set of ${\bf M}$. 
Notice that the definition of a circuit implies that, 
for every pair of distinct circuits $C_1,C_2$ of ${\bf M}$, 
we have 
$C_1 \setminus C_2 \neq \emptyset$ and 
$C_2 \setminus C_1 \neq \emptyset$.  

\begin{lemma}[See, e.g., {\cite[page 15, Exercise 14]{O11}}] \label{lemma:elimination}
Let $C_1,C_2$ be distinct circuits 
of ${\bf M}$ such that $C_1 \cap C_2 \neq \emptyset$.
Then for 
every element $u \in C_1 \cap C_2$
and every element $v \in C_1 \setminus C_2$, 
there exists a circuit $C$ of ${\bf M}$ 
such that $v \in C \subseteq (C_1 \cup C_2) - u$. 
\end{lemma}

Let $I$ be an independent set of ${\bf M}$.
Define 
${\sf sp}_{{\bf M}}(I)$ as the set of elements 
$u \in U \setminus I$ such that $I + u \notin \mathcal{I}$. 
Then it is not difficult to see that, 
for every element $u \in {\sf sp}_{{\bf M}}(I)$, 
$I + u$ contains a circuit of ${\bf M}$ as  
a subset, and (I1) implies that $u$ belongs to this circuit.
Notice that Lemma~\ref{lemma:elimination}
implies that such a circuit is uniquely determined. 
We call this circuit the {\it fundamental circuit of $u$ 
for $I$ and ${\bf M}$}, and we denote by 
${\sf C}_{{\bf M}}(u,I)$ this circuit.
It is well known that, 
for every element $u \in {\sf sp}_{{\bf M}}(I)$, 
${\sf C}_{{\bf M}}(u,I)$ is the set of 
elements $v \in I  + u$ such that  
$I + u - v \in \mathcal{I}$
(see, e.g., \cite[p.20, Exercise 5]{O11}). 
For each element $u \in {\sf sp}_{{\bf M}}(I)$, 
we define ${\sf D}_{{\bf M}}(u,I) := {\sf C}_{{\bf M}}(u,I) - u$. 
Then (H2) can be rewritten as follows. 
\begin{itemize}
\item
$e \in {\sf sp}_{\bf H}(\mu)$, and 
there exists an edge $f \in {\sf D}_{\bf H}(e,\mu)$ such that 
$e \succsim_H f$
(resp.\ $e \succ_H f$).
\end{itemize}

Let $X$ be a subset of $U$.
We define $\mathcal{I}|X$ as the family of 
subsets $ I\subseteq X$ such that $I \in \mathcal{I}$. 
Furthermore, we define ${\bf M}|X := (X, \mathcal{I}|X)$.
It is known that 
${\bf M}|X$ is a matroid
(see, e.g., \cite[p.20]{O11}). 
Define ${\sf r}_{{\bf M}}(X)$ as the size of a base of ${\bf M}|X$. 
For each subset $I \subseteq U \setminus X$, 
we define 
\begin{equation*}
{\sf p}_{{\bf M}}(I; X) := {\sf r}_{{\bf M}}(I \cup X) - {\sf r}_{{\bf M}}(X).
\end{equation*}
Furthermore, we define 
\begin{equation*}
\mathcal{I}/X := \{I \subseteq U \setminus X \mid {\sf p}_{{\bf M}}(I; X) = |I|\}, \ \ \ 
{\bf M}/X := (U \setminus X, \mathcal{I}/X).
\end{equation*}  
It is known that ${\bf M}/X$ is a matroid
(see, e.g., \cite[Proposition~3.1.6]{O11}).

\begin{lemma}[{See, e.g., \cite[Proposition~3.1.25]{O11}}] \label{lemma:minor}
For every pair of disjoint subsets $X,Y \subseteq U$, 
\begin{equation*}
({\bf M}/X) / Y = {\bf M}/ (X \cup Y), \ \ \ 
({\bf M}/X) | Y = ({\bf M}| (X \cup Y)) / X.
\end{equation*}  
\end{lemma}

\begin{lemma}[{See, e.g., \cite[Proposition~3.1.7]{O11}}] \label{lemma:contraction}
Let $X$ and $B$ be a subset of $U$ and 
a base of ${\bf M}|X$, respectively. 
Then for every subset $I \subseteq U \setminus X$, 
$I$ is an independent set 
{\rm (}resp.\ a base{\rm )}
of ${\bf M}/ X$
if and only if 
$I \cup B$ is an independent set 
{\rm (}resp.\ a base{\rm )}
of ${\bf M}$. 
\end{lemma} 

\begin{lemma}[{Iri and Tomizawa~\cite[Lemma~2]{IriT76}}] \label{lemma:IriT76_sequence}
Let $I$ be an independent set of ${\bf M}$.
We are given 
distinct elements $u_1,u_2,\ldots,u_{\ell} \in U \setminus I$ 
and $v_1,v_2,\ldots,v_{\ell} \in I$ satisfying the following 
conditions.
\begin{itemize}
\item
$u_i \in {\sf sp}_{\bf M}(I)$ 
for every integer $i \in [\ell]$.
\item
$v_i \in {\sf C}_{\bf M}(u_i,I)$
for every integer $i \in [\ell]$. 
\item
$v_i \notin {\sf C}_{\bf M}(u_j,I)$
holds for every pair of integers $i,j \in [\ell]$ such that $i < j$.
\end{itemize}
Define $J := (I \cup \{u_1,u_2,\ldots,u_{\ell}\}) \setminus \{v_1,v_2,\ldots,v_{\ell}\}$.
Then $J \in \mathcal{I}$ and 
${\sf sp}_{\bf M}(I) \cup I = {\sf sp}_{\bf M}(J) \cup J$. 
\end{lemma}

Although the following lemmas 
easily follow from 
Lemma~\ref{lemma:elimination},  
we give the proofs of these lemmas for reader's convenience.

\begin{lemma}[{See, e.g., \cite[Lemma~2.4]{K22}}] \label{lemma:circuit_union}
Let $C, C_1,C_2,\ldots,C_{\ell}$ be circuits of ${\bf M}$.
We assume that 
distinct elements $u_1,u_2,\ldots,u_{\ell}, v \in U$ satisfy the following 
conditions. 
\begin{itemize}
\item[\bf (U1)]
$u_i \in C \cap C_i$  
for every integer $i \in [\ell]$. 
\item[\bf (U2)]
$u_{i} \notin C_{j}$ holds for every 
pair of distinct integers $i,j \in [\ell]$. 
\item[\bf (U3)]
$v \in C \setminus (C_1 \cup C_2 \cup \cdots \cup C_{\ell})$.  
\end{itemize}
Then there exists a circuit $C^{\prime}$ of ${\bf M}$ 
such that 
\begin{equation*}
C^{\prime} \subseteq \big(C \cup C_1 \cup C_2 \cup \cdots \cup C_{\ell}\big) 
\setminus \{u_1,u_2,\ldots,u_{\ell}\}.
\end{equation*} 
\end{lemma}
\begin{proof}
In order to prove this lemma,
it suffices to prove that, for every integer $i \in \{0\} \cup [{\ell}]$, 
there exists a circuit $C^{\prime}$ of ${\bf M}$ 
such that 
\begin{equation} \label{eq1:lemma:circuit_union}
v \in C^{\prime} \subseteq \big(C \cup C_1 \cup C_2 \cup \dots \cup C_i\big) 
\setminus \{u_1,u_2,\ldots,u_i\} 
\tag{$\star$}
\end{equation}
by induction on $i$. 
When $i = 0$, the circuit $C$ satisfies 
\eqref{eq1:lemma:circuit_union}. 

Let $z$ be an integer in $[{\ell}]$. 
Assume that \eqref{eq1:lemma:circuit_union} holds
when $i = z -1$. 
Let $C^{\prime}$ be a circuit of ${\bf M}$ 
satisfying \eqref{eq1:lemma:circuit_union} when $i = z -1$. 
If $u_z \notin C^{\prime}$, then 
$C^{\prime}$ satisfies \eqref{eq1:lemma:circuit_union} when $i = z$. 
Thus, we assume that $u_z \in C^{\prime}$. 
In this case, $u_z \in C_z \cap C^{\prime}$ and
$v \in C^{\prime} \setminus C_z$. 
Thus, Lemma~\ref{lemma:elimination} implies that 
there exists a circuit $C^{\circ}$ of ${\bf M}$ such that 
$v \in C^{\circ} \subseteq (C^{\prime} \cup C_z) - u_z$.
Furthermore, (U2) implies that 
$u_1,u_2,\ldots,u_{z-1} \notin C^{\circ}$. 
These imply that 
$C^{\circ}$ satisfies \eqref{eq1:lemma:circuit_union} when 
$i = z$. This completes the proof. 
\end{proof} 

Lemma~\ref{lemma:IriT76_sequence} implies that, 
for every independent set $I$ of ${\bf M}$, 
every element $u \in {\sf sp}_{\bf M}(I)$,
every element $v \in {\sf C}_{\bf M}(u,I)$, and 
every element $x \in {\sf sp}_{\bf M}(I) - u$, 
we have $x \in {\sf sp}_{\bf M}(I + u - v)$. 

\begin{lemma}[{See, e.g., \cite[Lemma~3]{IriT76}}] \label{lemma:IriT76_auxiliary}
Let $I$ be an independent set of ${\bf M}$.
Let $u$ and $v$ 
be an element 
in ${\sf sp}_{\bf M}(I)$ and 
an element in ${\sf C}_{\bf M}(u,I)$, respectively. 
Define $J := I + u - v$. 
Furthermore, let $x$ and $y$ be 
an element in ${\sf sp}_{{\bf M}}(I) - u$
and an element in $I - v$, respectively. 
Then the following statements hold.
\begin{enumerate}
\item[(1)]
If $v \in {\sf C}_{\bf M}(x, I)$
{\rm (}resp.\ $v \notin {\sf C}_{\bf M}(x, I)${\rm )}, 
then $u \in {\sf C}_{\bf M}(x, J)$
{\rm (}resp.\ $u \notin {\sf C}_{\bf M}(x, J)${\rm )}. 
\item[(2)]
Assume that $y \in {\sf C}_{\bf M}(x,I)$
{\rm (}resp.\ $y \notin {\sf C}_{\bf M}(x,I)${\rm )}.
Then if at least one of 
$y \notin {\sf C}_{\bf M}(u,I)$ and 
$v \notin {\sf C}_{\bf M}(x,I)$ 
holds, then 
$y \in {\sf C}_{\bf M}(x,J)$
{\rm (}resp.\ $y \notin {\sf C}_{\bf M}(x,J)${\rm )}.
\end{enumerate}
\end{lemma} 
\begin{proof}
Define 
$C_{x,I} := {\sf C}_{\bf M}(x,I)$ and
$C_{x,J} := {\sf C}_{\bf M}(x,J)$. 
Then we prove (1).
Assume that 
$v \in C_{x,I}$ and $u \notin C_{x,J}$
(resp.\ $v \notin C_{x,I}$ and $u \in C_{x,J}$). 
Then $C_{x,J} - x \subseteq I$
(resp.\ $C_{x,I} - x \subseteq J$). 
Since 
$v \in C_{x,I}$ and 
$v \notin C_{x,J}$
(resp.\ $u \notin C_{x,I}$ and $u \in C_{x,J}$), 
$C_{x,I} \neq C_{x,J}$. 
Thus, since $x \in C_{x,I} \cap C_{x,J}$, 
Lemma~\ref{lemma:elimination} implies that 
there exists a circuit $C$ of ${\bf M}$ such that
$C \subseteq (C_{x,I} \cup C_{x,J}) - x \subseteq I$
(resp.\ $\subseteq J$). 
However, this contradicts the fact that $I \in \mathcal{I}$
(resp.\ $J \in \mathcal{I}$).

In addition, we define $C_{u,I} := {\sf C}_{\bf M}(u,I)$.
Then we prove (2). 
We first assume that 
$y \notin C_{u,I}$ and 
$y \notin C_{x,J}$
(resp.\ $y \in C_{x,J}$). 
Recall that 
$v \in C_{u,I}$. 
Since 
$y \in C_{x,I}$ 
and $y \notin C_{x,J}$
(resp.\ $y \notin C_{x,I}$ and $y \in C_{x,J}$), 
$C_{x,I} \neq C_{x,J}$. 
Thus, Lemma~\ref{lemma:elimination} implies that 
there exists a circuit $C$ of ${\bf M}$ such that
$C \subseteq (C_{x,I} \cup C_{x,J}) - x$. 
If $u \notin C_{x,J}$, then 
$C \subseteq I$. 
This contradicts the fact that $I \in \mathcal{I}$. 
Thus, $u \in C_{x,J}$. 
Since $v \in C_{u,I}$ and $v \notin C_{x,J}$,  
$C_{x,J} \neq C_{u,I}$. 
Thus, since $u \in C_{x,J} \cap C_{u,I}$ and 
$y \notin C_{x,J} \cup C_{u,I}$ 
(resp.\ $y \in C_{x,J} \setminus C_{u,I}$), 
Lemma~\ref{lemma:elimination} implies that 
there exists a circuit $C^{\prime}$ of ${\bf M}$ such that 
$v \in C^{\prime} \subseteq I + x$ and 
$y \notin C^{\prime}$
(resp.\ $y \in C^{\prime}$).
If $x \notin C^{\prime}$, then 
$C^{\prime} \subseteq I$. 
Thus, $x \in C^{\prime}$.  
Since $y \in C_{x,I}$ and 
$y \notin C^{\prime}$ 
(resp.\ $y \notin C_{x,I}$ and $y \in C^{\prime}$), 
$C_{x,I} \neq C^{\prime}$. 
Thus, since $x \in C_{x,I} \cap C^{\prime}$, 
it follows from Lemma~\ref{lemma:elimination} that 
there exists a circuit $C^{\circ}$ of ${\bf M}$ such that
$C^{\circ} \subseteq I$. 
This is a contradiction.  

Assume that $v \notin C_{x,I}$ 
and $y \notin C_{x,J}$ 
(resp.\ $y \in C_{x,J}$). 
Then $C_{x,I} \subseteq I + x - v$
holds. 
Since $y \in C_{x,I}$
(resp.\ $y \notin C_{x,I}$), 
$C_{x,I} \neq C_{x,J}$. 
Thus, Lemma~\ref{lemma:elimination} implies that 
there exists a circuit $C$ of ${\bf M}$ such that
$C \subseteq J$. This is a contradiction. 
Thus, this completes the proof of (2). 
\end{proof} 

\subsection{Additional facts on matroids} 

Let ${\bf M}_1 = (U, \mathcal{I}_1), {\bf M}_2 = (U, \mathcal{I}_2)$ 
be matroids. 
An element in $\mathcal{I}_1 \cap \mathcal{I}_2$ is called a 
\emph{common independent set of ${\bf M}_1$ and ${\bf M}_2$}. 
Then it is known that if, for every integer $i \in \{1,2\}$ and every subset $X \subseteq U$, 
we can determine whether $X \in \mathcal{I}_i$ in time bounded by a 
polynomial in $|U|$, then  
we can find a maximum-size 
common independent set of 
${\bf M}_1$ and ${\bf M}_2$ in time bounded by a 
polynomial in $|U|$ (see, e.g., \cite{C86,Bli21}) . 

Let ${\bf M}, {\bf N}$ 
be matroids defined on the same ground set $U$. 
Then for each common independent set $I$ of ${\bf M}$ and ${\bf N}$,
we define a simple directed graph 
${\sf G}_{{\bf M},{\bf N}}(I) = 
(U, {\sf A}_{{\bf M},{\bf N}}(I))$ 
as follows. 
\begin{itemize}
\item
For 
each element $u \in U \setminus I$
and 
each element $v \in I$, 
$vu \in {\sf A}_{{\bf M},{\bf N}}(I)$ 
if and only if $u \in {\sf sp}_{\bf M}(I)$ and 
$v \in {\sf D}_{{\bf M}}(u,I)$. 
\item
For 
each element $u \in U \setminus I$ 
and 
each element $v \in I$, 
$uv \in {\sf A}_{{\bf M},{\bf N}}(I)$ 
if and only if $u \in {\sf sp}_{\bf N}(I)$ and 
$v \in {\sf D}_{{\bf N}}(u,I)$. 
\end{itemize}
Let $I$ be  
a common independent set of ${\bf M}$ and ${\bf N}$, 
and let $P$ be a simple directed cycle in 
${\sf G}_{{\bf M},{\bf N}}(I)$. 
In addition, let $X$ (resp.\ $Y$) be 
the set of elements in $U \setminus I$ (resp.\ 
$I$) that $P$ passes through. 
Then we define 
$I \ominus P := (I \cup X) \setminus Y$. 

\begin{lemma} \label{lemma:IriT76_intersection}
Let ${\bf M}, {\bf N}$ 
be matroids defined on the same ground set, and  
let $I$ be a common independent set of ${\bf M}$ and ${\bf N}$. 
Furthermore, let $P$ be a simple directed cycle in 
${\sf G}_{{\bf M},{\bf N}}(I)$ such that 
there does not exist a shortcut arc in ${\sf A}_{{\bf M},{\bf N}}(I)$ for  
$P$. 
Then $I \ominus P$ is a common independent set of ${\bf M}$ and ${\bf N}$. 
\end{lemma} 
\begin{proof} 
This lemma immediately follows from Lemma~\ref{lemma:IriT76_sequence}. 
\end{proof} 

Let ${\bf M}_1 = (U_1, \mathcal{I}_1), {\bf M}_2 = (U_2, \mathcal{I}_2),\dots,
{\bf M}_z = (U_z, \mathcal{I}_z)$ 
be matroids such that 
$U_1,U_2,\dots,U_z$ are pairwise disjoint. 
Then the \emph{direct sum} ${\bf M} = (U, \mathcal{I})$ of 
${\bf M}_1, {\bf M}_2, \dots, {\bf M}_z$
is defined as follows. 
We define $U$ as $U_1 \cup U_2 \cup \dots \cup U_z$. 
We define $\mathcal{I}$ as the family of subsets 
$I \subseteq U$ such that 
$I \cap U_i \in \mathcal{I}_i$ for every integer $i \in [z]$. 
It is known that 
${\bf M}$ is a matroid
(see, e.g., \cite[Proposition~4.2.12]{O11}).

\section{Algorithm} 

For each doctor $d \in D$ and 
each subset $F \subseteq E(d)$, 
we define 
\begin{equation*}
{\sf head}_d(F) := \{e \in F \mid \mbox{$e \succsim_d f$
for every edge $f \in F$}\}. 
\end{equation*} 
For each subset $F \subseteq E$, we define 
\begin{equation*}
\begin{split}
{\sf head}_H(F) & := \{e \in F \mid \mbox{$e \succsim_H f$
for every edge $f \in F$}\},\\
{\sf tail}_H(F) & := \{e \in F \mid \mbox{$f \succsim_H e$
for every edge $f \in F$}\}. 
\end{split}
\end{equation*} 
Furthermore, for each matching $\mu$ in $G$, we define 
${\sf block}_H(\mu)$ as the set of 
edges $e = (d,h) \in {\sf sp}_{\bf H}(\mu)$
such that $e \succsim_d \mu(d)$ and 
one of the 
following conditions is satisfied.
\begin{itemize}
\item
$e \succ_d \mu(d)$, and 
there exists an edge $f \in {\sf D}_{\bf H}(e,\mu)$ such that 
$e \succsim_H f$.  
\item
$e \sim_d \mu(d)$, and 
there exists an edge $f \in {\sf D}_{\bf H}(e,\mu)$ such that 
$e \succ_H f$.  
\end{itemize}
For each subset $F \subseteq E$, 
we define the subset ${\rm Ch}_D(F) \subseteq F$ by 
\begin{equation*}
{\rm Ch}_D(F) := 
\bigcup_{d \in D} {\sf head}_d(F(d)). 
\end{equation*}

Define $\mathcal{U}$ as the family of subsets 
$F \subseteq E$ such that 
$|F(d)| \le 1$ for every doctor $d \in D$. 
Define the matroid ${\bf D}$ 
by ${\bf D} := (E, \mathcal{U})$.

\subsection{Subroutine} 

Throughout this subsection, let $F$ be a subset of 
$E$. 
Then we define the matroid ${\bf H}\langle F \rangle$ as the 
matroid defined by Algorithm~\ref{alg:sub_matroid}. 

\begin{algorithm}[h]
If $F = \emptyset$, then define 
${\bf H}\langle F \rangle := (\emptyset, \{\emptyset\})$ and 
$i_F := 0$, and halt.\\
Set $t := 0$.\\
Define $F_0 := F$ and ${\bf H}_0 := {\bf H}$. \\
\While{$F_t \neq \emptyset$}
{
    Set $t := t + 1$.\\
    Define $T_t := {\sf head}_H(F_{t-1})$ and $F_t := F_{t-1} \setminus T_t$.\\
    Define 
    ${\bf H}_t := {\bf H}_{t-1} / T_t$ and 
    ${\bf T}_t := {\bf H}_{t-1}|T_t$.\\
}
Define $i_F := t$.\\
Define ${\bf H}\langle F \rangle$ as 
the direct sum of ${\bf T}_1, {\bf T}_2, \dots,
{\bf T}_{i_F}$, and halt. 
\caption{Algorithm for defining ${\bf H}\langle F \rangle$}
\label{alg:sub_matroid}
\end{algorithm}

For each integer $i \in [i_F]$, we define 
$T_{1,i} := T_1 \cup T_2 \cup \dots \cup T_i$. 
In addition, we define $T_{1,0} := \emptyset$. 
Notice that $T_{1,i_F} = F$. 

Notice that, for every integer $i \in [i_F]$ and 
every pair of edges $e,f \in T_i$, 
we have $e \sim_H f$. 
Lemma~\ref{lemma:contraction} implies that, 
for every subset $I \subseteq F$, 
we can determine whether $I$ is an independent set of 
${\bf H}\langle F \rangle$ 
in polynomial time as follows.
Since ${\bf T}_1 = {\bf H}|T_1$, 
we can determine whether $I \cap T_1$ is an independent set
of ${\bf T}_1$ by checking whether $I \cap T_1 \in \mathcal{F}$. 
Let $i$ be an integer in $[i_F]$ such that $i \ge 2$.
Then we compute a base $B$ of ${\bf H}|T_{1,i-1}
= ({\bf H}|T_{1,i})|T_{1,i-1}$ in polynomial time.
Lemma~\ref{lemma:minor} implies that
${\bf T}_i = ({\bf H}/T_{1,i-1})|T_i = 
({\bf H}|T_{1,i})/T_{1,i-1}$.
Thus, it follows from Lemma~\ref{lemma:contraction} that 
we can determine whether $I \cap T_i$ is an independent 
set of ${\bf T}_i$ by checking 
whether $(I \cap T_i) \cup B \in \mathcal{F}$. 
Then we do this for each integer $i \in [i_F]$. 

\begin{lemma} \label{lemma:new_matroid_independent}
Let $B$ be a base of ${\bf H}\langle F \rangle$.
Then for every integer $i \in [i_F]$, 
$B \cap T_{1,i}$ is a base of ${\bf H}|T_{1,i}$.
Especially, $B$ is a base of ${\bf H}|F$.
\end{lemma}
\begin{proof}
For each integer $i \in [i_F]$, we define 
$B_i := B \cap T_{1,i}$. 
Since $B$ is a base of ${\bf H}\langle F \rangle$, 
$B \cap T_i$ is a base of ${\bf T}_i$
for every integer $i \in [i_F]$. 
Thus, since ${\bf T}_1 = {\bf H}| T_{1,1}$, 
$B_1$ is a base of ${\bf H}|T_{1,1}$. 

Assume that $i_F \ge 2$. 
Let $i$ be an integer in $[i_F-1]$.  
In addition, we assume that $B_i$ is a 
base of ${\bf H}|T_{1,i} = ({\bf H}|T_{1,i+1})|T_{1,i}$. 
Recall that 
$B \cap T_{i+1}$ is a base of ${\bf T}_{i+1}$.
Furthermore, 
Lemma~\ref{lemma:minor} implies that 
\begin{equation*}
{\bf T}_{i+1} = {\bf H}_i | T_{i+1} 
= ({\bf H} / T_{1,i}) |T_{i+1}
= ({\bf H} | T_{1,i+1}) /T_{1,i}. 
\end{equation*}
Thus, 
Lemma~\ref{lemma:contraction} 
implies that 
$B_{i} \cup (B \cap T_{i+1}) = B_{i+1}$ is a
base 
of ${\bf H}|T_{1,i+1}$. 
This completes the proof. 
\end{proof} 

Lemma~\ref{lemma:new_matroid_independent} implies that 
${\sf r}_{\bf H}(F) = {\sf r}_{{\bf H}|F}(F) = {\sf r}_{{\bf H}\langle F \rangle}(F)$.
In addition, 
Lemma~\ref{lemma:new_matroid_independent} implies that,
for every base $B$ of ${\bf H}\langle F \rangle$ and 
every edge $e \in F \setminus B$, 
we have $e \in {\sf sp}_{{\bf H}|F}(B)$ and 
$e \in {\sf sp}_{\bf H}(B)$. 

\begin{lemma} \label{lemma:new_matroid_circuit}
Let $B$ be a base of ${\bf H}\langle F \rangle$, and 
let $e$ be an edge in $F \setminus B$.  
Then there exists an integer $z \in [i_F]$ 
such that 
${\sf C}_{{\bf H}\langle F \rangle}(e,B)$ is 
a circuit of ${\bf T}_z$. 
Furthermore, 
${\sf C}_{{\bf H}\langle F \rangle}(e,B) \subseteq 
{\sf C}_{\bf H}(e,B)$, and 
$g \succ_H f$ for every pair of edges 
$f \in {\sf C}_{{\bf H}\langle F \rangle}(e,B)$ and 
$g \in {\sf C}_{\bf H}(e,B) \setminus {\sf C}_{{\bf H}\langle F \rangle}(e,B)$. 
\end{lemma}
\begin{proof}
The first statement follows from the fact that 
${\bf H}\langle F \rangle$ is the direct sum of 
${\bf T}_1,{\bf T}_2,\dots,{\bf T}_{i_F}$. 
In what follows, let $z$ be the integer satisfying the condition of the first
statement. 

Let $f$ be an edge in ${\sf C}_{{\bf H}\langle F \rangle}(e,B)$, and 
we define $B^{\prime} := B + e - f$. Then 
$B^{\prime}$ is a base of ${\bf H}\langle F \rangle$. 
Thus, Lemma~\ref{lemma:new_matroid_independent} implies that $B^{\prime}$ 
is an independent set of ${\bf H}$. 
This implies that $f \in {\sf C}_{\bf H}(e,B)$. 

Next, we assume that there exists an edge $f \in {\sf C}_{\bf H}(e,B)$
such that $f \in T_z \setminus {\sf C}_{{\bf H}\langle F \rangle}(e,B)$
or $g \succ_H f$ for an edge $g \in T_z$, and we 
define $B^{\prime} := B + e - f$.
Then $B^{\prime}$ is an independent set 
of ${\bf H}$. 
Lemma~\ref{lemma:new_matroid_independent}
implies that 
$B^{\prime} \cap T_{1,z-1}$ ($= B \cap T_{1,z-1}$) is a base of 
${\bf H}|T_{1,z-1}$. 
Thus, Lemma~\ref{lemma:contraction} implies 
that $B^{\prime} \cap T_z$ is an independent set of 
${\bf T}_z$. 
However, ${\sf C}_{{\bf H}\langle F \rangle}(e,B) \subseteq B^{\prime} \cap T_z$.
This is a contradiction.
\end{proof} 

\subsection{Description of the algorithm} 

For each subset $F \subseteq E$, we define 
the subset $D[F] \subseteq D$ as the set of doctors $d \in D$ 
such that 
$F(d) \neq \emptyset$. 
For each subset $F \subseteq E$, we define 
the function $\rho_F \colon 2^{D[F]} \to \mathbb{Z}$ 
by 
\begin{equation*}
\rho_F(X) := {\sf r}_{{\bf H} \langle F \rangle}(F(X)) - |X|. 
\end{equation*}
This kind of function was also used in \cite{K21}. 

\begin{lemma}[Rado~\cite{R42}] \label{lemma:R42}
Let $F$ be a subset of $E$.
Then there 
exists a common independent set $I$ of 
${\bf D}|F$ and ${\bf H}\langle F \rangle$ such that 
$|I| = |D[F]|$
if and only if 
$\rho_F(X) \ge 0$
for every subset $X \subseteq D[F]$. 
\end{lemma}

Although the following lemma easily follows from the submodularity of 
${\sf r}_{{\bf H} \langle F \rangle}$, we give a proof for completeness. 

\begin{lemma} \label{lemma:submodular_rho}
For every subset $F \subseteq E$, 
$\rho_F$ is submodular. 
\end{lemma}
\begin{proof}
Let $F$ be a subset of $E$. 
For notational simplicity, 
we define ${\sf r}^{\ast} := {\sf r}_{{\bf H} \langle F \rangle}$. 
Then it is known that 
${\sf r}^{\ast}$ is submodular
(see, e.g., \cite[Lemma~1.3.1]{O11}). 
Thus, for every pair of subsets $X,Y \subseteq D[F]$, 
\begin{equation*}
\begin{split}
& \rho_F(X) + \rho_F(Y) 
= {\sf r}^{\ast}(F(X)) + {\sf r}^{\ast}(F(Y)) 
- (|X| + |Y|) \\
& \ge 
{\sf r}^{\ast}(F(X) \cup F(Y)) + {\sf r}^{\ast}(F(X) \cap F(Y)) 
- (|X \cup Y| + |X \cap Y|) \\
& = 
{\sf r}^{\ast}(F(X \cup Y)) + {\sf r}^{\ast}(F(X \cap Y)) 
- (|X \cup Y| + |X \cap Y|) \\
& = \rho_F(X \cup Y) + \rho_F(X \cap Y).
\end{split} 
\end{equation*}
This completes the proof. 
\end{proof} 

We are now ready to propose our algorithm, which 
is described in Algorithm~\ref{alg:main}. 
See Section~\ref{section:example} 
for a concrete example of execution of 
Algorithm~\ref{alg:main}.  

\begin{algorithm}[h]
Set $t := 0$.
Define ${\sf R}_0 := \emptyset$.\\
\Do{${\sf R}_t \neq {\sf P}_{t,i_t}$}
{
    Set $t := t+1$ and $i := 0$.\\
    Define ${\sf P}_{t,0} := {\sf R}_{t-1}$.\\
    \Do{${\sf P}_{t,i} \neq {\sf P}_{t,i-1}$}
    {
        Set $i := i + 1$.\\
        Define $K_{t,i} := {\rm Ch}_D(E \setminus {\sf P}_{t,i-1})$.\\
        \If{${\sf r}_{\bf H}(K_{t,i}) > |D[E \setminus {\sf P}_{t,i-1}]|$}
        {
            Define ${\sf R}_t := {\sf P}_{t,i-1}$ and $i_t := i -1$.\\
            Output {\bf null}, and halt. 
        }
        Find a maximum-size common independent set $\kappa_{t,i}$ 
        of ${\bf D}|K_{t,i}$ and ${\bf H}\langle K_{t,i} \rangle$.\\
        \uIf{$|\kappa_{t,i}| < |D[E \setminus {\sf P}_{t,i-1}]|$}
        {
             Find the minimal minimizer $Z_{t,i}$
             of $\rho_{K_{t,i}}$. \\
             Define ${\sf P}_{t,i} := {\sf P}_{t,i-1} \cup K_{t,i}(Z_{t,i})$.
        }
        \Else
        {
            Define ${\sf P}_{t,i} := {\sf P}_{t,i-1}$.
        }
    }
    Define $i_t := i$.\\
    \uIf{${\sf P}_{t,i_t} \cap {\sf block}_H(\kappa_{t,i_t}) \neq \emptyset$}
    {
        Define $b_{t}$ as an edge in 
        ${\sf P}_{t,i_t} \cap {\sf block}_H(\kappa_{t,i_t})$.\\
        Define 
        ${\sf R}_{t} := {\sf P}_{t,i_t} \cup {\sf tail}_H({\sf C}_{\bf H}(b_{t}, \kappa_{t,i_t}))$.\\
    }
    \Else
    {
        Define ${\sf R}_{t} := {\sf P}_{t,i_t}$. 
    }
}
Define $n := t$ and $\mu_{\rm o} := \kappa_{n,i_n}$.\\
\uIf{there exists an edge $e_{\rm R} \in {\sf R}_n$ 
such that $\mu_{\rm o} + e_{\rm R} \in \mathcal{F}$}
{
   Output ${\bf null}$, and halt. 
}
\Else
{
   Output $\mu_{\rm o}$, and halt.
}
\caption{Proposed algorithm}
\label{alg:main}
\end{algorithm}

Lemma~\ref{lemma:R42} implies that, in Step~14, 
since there exists a subset $X \subseteq D[K_{t,i}]$ such that 
$\rho_{K_{t,i}}(X) < 0$, we have $Z_{t,i} \neq \emptyset$. 
Thus, in the course of Algorithm~\ref{alg:main}, 
if $|\kappa_{t,i}| < |D[E \setminus {\sf P}_{t,i-1}]|$, then 
${\sf P}_{t,i-1} \subsetneq {\sf P}_{t,i}$. 
If
${\sf P}_{t,i_t} \cap {\sf block}_H(\kappa_{t,i_t}) \neq \emptyset$, then 
since ${\sf tail}_H({\sf C}_{\bf H}(b_{t}, \kappa_{t,i_t})) \neq \{b_t\}$, 
we have
${\sf R}_{t-1} \subsetneq {\sf R}_t$. 
In Step~21, if we define 
$r_{t,i_t} := {\sf r}_{{\bf H}\langle K_{t,i_t} \rangle}(K_{t,i_t})$, 
then Lemma~\ref{lemma:new_matroid_independent} implies that 
\begin{equation*}
r_{t,i_t}
\ge 
|\kappa_{t,i_t}| \ge |D[E \setminus {\sf P}_{t,i_t-1}]| 
\ge {\sf r}_{\bf H}(K_{t,i_t}) 
= r_{t,i_t}.
\end{equation*} 
Thus, $|\kappa_{t,i_t}| = |D[E \setminus {\sf P}_{t,i_t-1}]|$
and $\kappa_{t,i_t}$ is a base of ${\bf H}\langle K_{t,i_t} \rangle$.

We prove the correctness of Algorithm~\ref{alg:main} in 
Section~\ref{section:correctness}. 
Here we prove that 
Algorithm~\ref{alg:main} is a polynomial-time 
algorithm. 
This follows from the following lemma. 
Since, for every subset $F \subseteq E$, we can 
determine whether $F \in \mathcal{F}$ in time bounded by a polynomial 
in the input size of $G$, 
for every subset $F \subseteq E$ and every subset $X \subseteq D[F]$, 
${\sf r}_{{\bf H} \langle F \rangle}(F(X))$ 
can be evaluated in polynomial time
by finding a base of ${\bf H} \langle F \rangle|F(X)$.
Thus, Step~14 can be done in polynomial time 
(see, e.g., \cite[Note~10.12]{M03}). 

\begin{lemma} \label{lemma:iteration}
The following statements hold.
\begin{enumerate}
\item[(1)]
In each iteration of Steps~2 to 27, 
the number of iterations of Steps~5 to 19 is $O(|E|)$. 
\item[(2)]
The number of iterations of Steps~2 to 27 is $O(|E|)$. 
\end{enumerate}
\end{lemma}
\begin{proof}
(1) 
If Algorithm~\ref{alg:main} does not go to 
Step~20 when $i = z$, then we have 
${\sf P}_{t,z-1} \subsetneq {\sf P}_{t,z}$. 
This completes the proof. 

(2)
If Algorithm~\ref{alg:main} does not go to 
Step~28 when $t = z$, then we have 
${\sf R}_{z-1} \subsetneq {\sf R}_{z}$. 
This completes the proof. 
\end{proof}

\section{Correctness} 
\label{section:correctness} 

In this section, we prove the following theorem. 

\begin{theorem} \label{theorem:main}
If Algorithm~\ref{alg:main} outputs {\bf null}, then 
there does not exist a strongly stable matching in $G$. 
Otherwise, the output of Algorithm~\ref{alg:main} is 
a strongly stable matching in $G$. 
\end{theorem}

In what follows, we prove lemmas needed to prove 
Theorem~\ref{theorem:main}. 

\begin{lemma} \label{lemma:output}
If Algorithm~\ref{alg:main} 
does not output {\bf null}, then 
the output $\mu_{\rm o}$ of Algorithm~\ref{alg:main} is 
a strongly stable matching in $G$. 
\end{lemma}
\begin{proof}
In order to prove this lemma by contradiction, 
we assume that $\mu_{\rm o}$ 
is not a strongly stable matching in $G$. 
For notational simplicity, 
we define $K := K_{n,i_n}$.
Recall that $\mu_{\rm o}$ is a common independent set of 
${\bf D}|K$ and ${\bf H}\langle K \rangle$.
Since $\mu_{\rm o}$ is a base of ${\bf H}\langle K \rangle$, Lemma~\ref{lemma:new_matroid_independent}
implies that $\mu_{\rm o}$ is a matching in $G$. 
This implies that 
$\mu_{\rm o}$ is not strongly stable, and 
there exists a blocking edge 
$e = (d,h) \in E \setminus \mu_{\rm o}$ 
for $\mu_{\rm o}$. 

First, we assume that 
$e \succ_d \mu_{\rm o}(d)$. 
Since $\mu_{\rm o} \subseteq K$, 
we have $e \in {\sf P}_{n,i_n}$. 
Since Algorithm~\ref{alg:main} does not 
output {\bf null} in Step~30, 
$e \in {\sf P}_{n,i_n} \cap {\sf block}_H(\mu_{\rm o})$. 
This contradicts the fact that 
${\sf R}_n = {\sf P}_{n,i_n}$.

Next, we assume that 
$e \sim_d \mu_{\rm o}(d)$. 
If $e \not \in K$, then 
$e \in {\sf P}_{n,i_n}$.
Since Algorithm~\ref{alg:main} does not 
output {\bf null} in Step~30, 
$e \in {\sf P}_{n,i_n} \cap {\sf block}_H(\mu_{\rm o})$. 
This contradicts the fact that 
${\sf R}_n = {\sf P}_{n,i_n}$.
Thus, we can assume that $e \in K$. 
Since $\mu_{\rm o}$ is a base of ${\bf H}\langle K \rangle$, 
$e \in {\sf sp}_{{\bf H}\langle K \rangle}(\mu_{\rm o})$.  
Thus, Lemma~\ref{lemma:new_matroid_circuit} 
implies that 
$f \succsim_H e$  
for every edge $f \in {\sf D}_{\bf H}(e, \mu_{\rm o})$. 
This contradicts the fact that 
$e$ is a blocking edge for $\mu_{\rm o}$.
\end{proof}

\begin{lemma} \label{lemma:reject}
Assume that Algorithm~\ref{alg:main} halts when 
$t = k$. Then 
for every strongly stable matching $\sigma$ in $G$, 
we have $\sigma \cap {\sf R}_k = \emptyset$. 
\end{lemma} 

We prove Lemma~\ref{lemma:reject} in Section~\ref{section:lemma:reject}.

\begin{lemma} \label{lemma:null_1}
If Algorithm~\ref{alg:main} 
outputs {\bf null} in Step~10, then 
there does not exist a strongly stable matching in $G$.
\end{lemma}
\begin{proof}
Assume that 
Algorithm~\ref{alg:main} 
outputs {\bf null} in Step~10 
when $t = k$ and $i = z$.
In addition, we assume that 
there exists a strongly stable matching 
$\sigma$ in $G$.
Then ${\sf r}_{\bf H}(K_{k,z}) > |D[E \setminus {\sf R}_{k}]|$.
Thus, since 
Lemma~\ref{lemma:reject} implies that 
$\sigma \subseteq E \setminus {\sf R}_{k}$, 
there exists an independent set 
$I$ of ${\bf H}$ such that 
$I \subseteq K_{k,z}$ and 
$|I| > |\sigma|$.
Since $\sigma \in \mathcal{F}$, 
(I2) implies that 
there exists an edge $e = (d,h) \in I \setminus \sigma$ such that 
$\sigma + e \in \mathcal{F}$. 
It follows from Lemma~\ref{lemma:reject} 
that 
$e \succsim_d \sigma(d)$. 
However, this contradicts the fact that 
$\sigma$ is strongly stable. 
\end{proof}

\begin{lemma} \label{lemma:null_2}
If Algorithm~\ref{alg:main} 
outputs {\bf null} in Step~30, then 
there does not exist a strongly stable matching in $G$.
\end{lemma}
\begin{proof}
Assume that 
Algorithm~\ref{alg:main} 
outputs {\bf null} in Step~30, and 
there exists a strongly stable matching 
$\sigma$ in $G$. 
Lemma~\ref{lemma:reject} implies that 
$\sigma \subseteq E \setminus {\sf R}_n$. 
Notice that $|\mu_{\rm o}| = |D[E \setminus {\sf R}_n]|$. 
Thus, $|\sigma| \le |\mu_{\rm o}|$. 
We define $\mu_+ := \mu_{\rm o} + e_{\rm R}$. 
Then since $\mu_+ \subseteq K_{n,i_n} \cup {\sf R}_n$
and $\mu_+ \in \mathcal{F}$,
there exists an edge $e = (d,h) \in \mu_+ \setminus \sigma$ such that 
$\sigma +  e \in \mathcal{F}$ and 
$e \succsim_d \sigma(d)$.  
This contradicts the fact that 
$\sigma$ is strongly stable. 
\end{proof}

\begin{proof}[Proof of Theorem~\ref{theorem:main}]
This theorem follows from Lemmas~\ref{lemma:output},  
\ref{lemma:null_1}, and \ref{lemma:null_2}. 
\end{proof} 

\subsection{Proof of Lemma~\ref{lemma:reject}}
\label{section:lemma:reject}

In this subsection, we prove Lemma~\ref{lemma:reject}. 
An edge $e \in {\sf R}_k$ is called a \emph{bad edge} if 
there exists a strongly stable matching $\sigma$ in $G$ such that 
$e \in \sigma$. 
If there does not exist a bad edge in ${\sf R}_k$, then 
the proof is done. 
Thus, we assume that there exists a bad edge in ${\sf R}_k$. 
We denote by $\Delta$ the set of integers 
$c \in [k]$ such that 
there exists a bad edge in ${\sf R}_c \setminus {\sf R}_{c-1}$. 
Define $c$ as the minimum integer in $\Delta$. 
We divide the proof into the following two cases. 
\begin{description}
\item[Case~A.] 
There exists an integer $z \in [i_c]$
such that 
${\sf P}_{c,z} \setminus {\sf P}_{c,z-1}$ 
contains a bad edge. 
\item[Case~B.]
${\sf P}_{c,i_c} \setminus {\sf R}_{c-1}$
does not 
contain a bad edge. 
\end{description} 

\subsubsection{Case~A}

First, we assume that 
there exists an integer $z \in [i_c]$
such that 
${\sf P}_{c,z} \setminus {\sf P}_{c,z-1}$ 
contains a bad edge.
Furthermore, we assume that $z$ is 
the minimum integer in $[i_c]$ 
satisfying this condition. 
In this case, there does not exist a bad edge in 
${\sf P}_{c,z-1}$. 
Let $\xi$ be a bad edge in ${\sf P}_{c,z} \setminus {\sf P}_{c,z-1}$. 
Then $\xi \in K_{c,z}(Z_{c,z})$. 
Let $\sigma$ be a strongly stable matching in $G$ 
such that $\xi \in \sigma$.
For notational simplicity, 
we define $K := K_{c,z}$, 
$Z := Z_{c,z}$, and ${\sf r}^{\ast} := {\sf r}_{{\bf H} \langle K \rangle}$. 

Define $L$ as the set of doctors $d \in Z$ such that 
$\sigma(d) \in K(Z)$. 
Then since $\xi \in \sigma \cap K(Z)$, 
we have $L \neq \emptyset$. 

\begin{lemma} \label{lemma_1:case_2}
For every doctor $d \in Z \setminus L$ and 
every edge $e \in K(d)$, we have 
$e \succ_d \sigma(d)$. 
\end{lemma}
\begin{proof}
Assume that there exist a 
doctor $d \in Z \setminus L$ and 
an edge $e \in K(d)$ such that 
$\sigma(d) \succsim_d e$. 
Since $d \notin L$, 
$\sigma(d) \notin K(d)$. 
Thus, since $K(d) = {\sf head}_d(E(d) \setminus {\sf P}_{c,z-1})$ 
and $e \in K(d)$,  
we have 
$\sigma(d) \in {\sf P}_{c,z-1}$. 
This implies that 
$\sigma(d)$ is a bad edge in ${\sf P}_{c,z-1}$.
However, this is a contradiction. 
This completes the proof.
\end{proof} 

In what follows, we prove that 
$Z \setminus L$ is a minimizer of 
$\rho_{K}$. 
This contradicts the fact that 
$Z$ is the inclusion-wise minimal 
minimizer of 
$\rho_{K}$.   
To this end,  
it is sufficient to prove that 
\begin{equation*}
\begin{split}
{\sf r}^{\ast}(K(Z)) 
- 
{\sf r}^{\ast}(K(Z \setminus L))
\ge |L|. 
\end{split}
\end{equation*}

Let $B$ be a base of 
${\bf H} \langle K \rangle | K(Z \setminus L)$. 
If 
$|B| \le {\sf r}^{\ast}(K(Z)) - |L|$, then 
the proof is done. 
Thus, we assume that 
$|B| > {\sf r}^{\ast}(K(Z)) - |L|$. 
Since $\sigma(L) \subseteq K(Z)$, 
we can obtain an independent set $I$ of 
${\bf H} \langle K \rangle | K(Z)$ 
satisfying the following conditions by greedily adding 
edges in $\sigma(L)$ to $B$. 
\begin{itemize}
\item
$B \subseteq I \subseteq B \cup \sigma(L)$.
\item
For every edge $e \in \sigma(L) \setminus I$, 
we have $e \in {\sf sp}_{{\bf H} \langle K \rangle | K(Z)}(I)$. 
\end{itemize}
Notice that the above assumption implies that 
$\sigma(L) \setminus I \neq \emptyset$. 
Let $\eta$ be an edge in $\sigma(L) \setminus I$. 
Define $C := {\sf C}_{{\bf H} \langle K \rangle| K(Z)}(\eta,I)$.
For every edge $e \in C$, 
$e \sim_H \eta$. 
Let $B^{\ast}$ be a base of ${\bf H} \langle K \rangle$
such that $I \subseteq B^{\ast}$.
Notice that (I2) guarantees the existence of 
such a base $B^{\ast}$. 
Then $\eta \notin B^{\ast}$, 
$\eta \in {\sf sp}_{{\bf H} \langle K \rangle}(B^{\ast})$, and 
Lemma~\ref{lemma:elimination} implies that 
${\sf C}_{{\bf H} \langle K \rangle}(\eta,B^{\ast}) = C$. 
We define $C^{\prime} := {\sf C}_{{\bf H}}(\eta,B^{\ast})$.
Then 
Lemma~\ref{lemma:new_matroid_circuit} implies that 
$f \succ_H e$ 
for every edge $e \in C$ and 
$f \in C^{\prime} \setminus C$. 
We prove that, 
for every edge $e \in C^{\prime} \setminus \sigma$, there 
exists a circuit $C_e$ of ${\bf H}$ such that 
$C_e - e \subseteq \sigma$ and 
$\eta \notin C_e$. 
If we can prove this, then Lemma~\ref{lemma:circuit_union}
implies that there exists a circuit $C^{\circ}$ of ${\bf H}$ 
such that
$C^{\circ} \subseteq \sigma$. 
This contradicts the fact that 
$\sigma \in \mathcal{F}$. 
This completes the proof. 

Let $e = (d,h)$ be an edge in $C \setminus \sigma$. 
Then $e \in B$ holds. 
Thus, since $B \subseteq K(Z \setminus L)$, 
$d \in Z \setminus L$. 
Lemma~\ref{lemma_1:case_2} implies that 
$e \succ_d \sigma(d)$.
Thus, since $\sigma$ is strongly stable, 
$e \in {\sf sp}_{\bf H}(\sigma)$ and 
$f \succ_H e$ for every edge $f \in {\sf D}_{\bf H}(e,\sigma)$. 
Since $\eta \sim_H e$ follows from the fact that 
$e \in C$, we have $\eta \notin {\sf D}_{\bf H}(e,\sigma)$. 
Thus, we can define $C_e := {\sf C}_{\bf H}(e,\sigma)$. 

Let $e = (d,h)$ be an edge in $C^{\prime} \setminus (C \cup \sigma)$.
Then since 
there does not exist a bad edge in ${\sf P}_{c,z-1}$, 
it follows from 
$e \in K$ that $e \succsim_d \sigma(d)$. 
Thus, since $\sigma$ is strongly stable, 
$e \in {\sf sp}_{\bf H}(\sigma)$ and 
$f \succsim_H e$ for every edge $f \in {\sf D}_{\bf H}(e,\sigma)$. 
Since $e \succ_d \eta$, we have $\eta \notin {\sf D}_{\bf H}(e,\sigma)$. 
Thus, we can define $C_e := {\sf C}_{\bf H}(e,\sigma)$. 
This completes the proof.

\subsubsection{Case~B} 

Next, we assume that 
${\sf P}_{c,i_c} \setminus {\sf R}_{c-1}$
does not contain a bad edge. 
Then ${\sf P}_{c,i_c}$ does not contain a bad edge. 
For notational simplicity, 
we define $b := b_c$,  
$\kappa := \kappa_{c,i_c}$, and 
$K := K_{c,i_c}$. 
In this case, ${\sf tail}_H({\sf C}_{\bf H}(b, \kappa)) - b$
contains a bad edge.
Let $\xi$ be a bad edge in ${\sf tail}_H({\sf C}_{\bf H}(b, \kappa)) - b$.
Let $\sigma$ be a strongly stable matching in $G$ 
such that $\xi \in \sigma$. 
Since $|\kappa| = {\sf r}_{{\bf H}\langle K \rangle}(K)$, 
$\kappa$ is a base of 
${\bf H}\langle K \rangle$.

\begin{lemma} \label{lemma_1:case_2}
$b \in {\sf sp}_{\bf H}(\sigma)$, and 
$e \succ_H \xi$ for 
every edge $e \in {\sf C}_{\bf H}(b, \sigma)$. 
\end{lemma}
\begin{proof}
Assume that $b = (d, h)$.
Then since $\xi \in {\sf tail}_H({\sf C}_{\bf H}(b, \kappa))$, 
we have $b \succsim_H \xi$. 

Assume that $b \sim_H \xi$.
Then for every edge 
$e \in {\sf C}_{\bf H}(b, \kappa)$, 
$e \succsim_H b$. Thus, since 
$b \in {\sf block}_H(\kappa)$, 
$b \succ_{d} \kappa(d)$. 
Furthermore, since ${\sf P}_{c,i_c}$ does not contain a bad edge,
$b \notin \sigma$ and $\kappa(d) \succsim_d \sigma(d)$. 
Since 
$b \succ_{d} \kappa(d)$, 
$b \succ_{d} \sigma(d)$. 
Thus, since $\sigma$ is strongly stable, 
$b \in {\sf sp}_{\bf H}(\sigma)$ and 
$e \succ_H b$ for every edge $e \in {\sf D}_{\bf H}(b,\sigma)$. 
Since $b \sim_H \xi$, 
$e \succ_H \xi$ for 
every edge $e \in {\sf C}_{\bf H}(b, \sigma)$. 

Assume that $b \succ_H \xi$.
Since 
$b \in {\sf block}_H(\kappa)$, 
$b \succsim_{d} \kappa(d)$.
Since 
${\sf P}_{c,i_c}$ does not contain a bad edge, 
$b \notin \sigma$ and $b \succsim_d \kappa(d) \succsim_d \sigma(d)$. 
Since $\sigma$ is strongly stable, 
$b \in {\sf sp}_{\bf H}(\sigma)$ and 
$e \succsim_H b$ for every edge $e \in {\sf D}_{\bf H}(b,\sigma)$. 
Since $b \succ_H \xi$, 
$e \succ_H \xi$ for 
every edge $e \in {\sf C}_{\bf H}(b, \sigma)$.  
\end{proof}

In this proof, the following lemma 
plays an important role. 
Notice that, in the following lemma, 
since $|\mu| = |\kappa|$, 
$\mu$ is a base of ${\bf H}\langle K \rangle$. 

\begin{lemma} \label{lemma_2:case_2}
There exists a common independent set 
$\mu$ of ${\bf D}|K$ and 
${\bf H}\langle K \rangle$ such that 
$|\mu| = |\kappa|$, 
$\mu(d) \succsim_d \kappa(d)$ for every doctor $d \in D$, 
and 
$\mu$ satisfies the following conditions.
\begin{itemize}
\item[\bf (P1)]
$b \in {\sf sp}_{\bf H}(\mu)$. 
\item[\bf (P2)]
$\xi \succsim_H e$ for every edge $e \in {\sf tail}_H({\sf C}_{\bf H}(b,\mu))$. 
\item[\bf (P3)]
$\sigma \cap {\sf tail}_H({\sf C}_{\bf H}(b,\mu)) \neq \emptyset$. 
\item[\bf (P4)] 
For every edge $e = (d,h) \in \mu \setminus \sigma$, 
we have $e \succ_d \sigma(d)$. 
\end{itemize}
\end{lemma} 

Before proving Lemma~\ref{lemma_2:case_2}, we prove this case 
with Lemma~\ref{lemma_2:case_2}. 
Assume that there exists a common independent set 
$\mu$ of ${\bf D}|K$ and 
${\bf H}\langle K \rangle$ satisfying 
the conditions in Lemma~\ref{lemma_2:case_2}.
Define $C := {\sf C}_{\bf H}(b,\mu)$.
Let $\eta$ be an edge in 
$\sigma \cap {\sf tail}_H(C)$. 
For every edge $e \in C$, 
we have $e \succsim_H \eta$. 
If we can prove that, for every edge $e \in C \setminus \sigma$, 
there exists a circuit $C_e$ of ${\bf H}$ such that 
$\eta \notin C_e$ and $C_e - e \subseteq \sigma$, 
then Lemma~\ref{lemma:circuit_union} 
implies that there exists a circuit $C^{\circ}$
of ${\bf H}$ such that 
$C^{\circ} \subseteq \sigma$. 
This contradicts the fact that 
$\sigma \in \mathcal{F}$. 
This completes the proof. 

Define $C_b := {\sf C}_{\bf H}(b,\sigma)$. 
Since $\xi \succsim_H \eta$, 
Lemma~\ref{lemma_1:case_2} implies that 
$\eta \notin C_b$. 
Let $e = (d,h)$ be an edge in $C \setminus (\sigma + b)$. 
Then (P4) implies that 
$e \succ_H \sigma(d)$. 
Thus, since $\sigma$ is strongly stable, 
$e \in {\sf sp}_{\bf H}(\sigma)$ and 
$f \succ_H e$ for every edge $f \in {\sf D}_{\bf H}(e, \sigma)$. 
Since $e \succsim_H \eta$, 
$\eta \notin {\sf D}_{\bf H}(e, \sigma)$. 
Thus, we can define $C_e := {\sf C}_{\bf H}(e,\sigma)$. 
This completes the proof. 

\begin{proof}[Proof of Lemma~\ref{lemma_2:case_2}]
Define $\Phi$ as the set of 
common independent sets 
$\mu$ of ${\bf D}|K$ and 
${\bf H}\langle K \rangle$
such that $|\mu| = |\kappa|$,
$\mu(d) \succsim_d \kappa(d)$ for every doctor $d \in D$, and 
$\mu$ satisfies (P1), (P2), and (P3).
Since $\kappa \in \Phi$, 
$\Phi \neq \emptyset$. 
For each element $\mu \in \Phi$, 
we define ${\sf def}(\mu)$ as the 
number of edges $e = (d,h) \in \mu \setminus \sigma$
such that 
$e \not \succ_d \sigma(d)$. 
In what follows, 
let $\mu$ be an element in $\Phi$
minimizing ${\sf def}(\mu)$ 
among all the elements in $\Phi$. 
If ${\sf def}(\mu) = 0$, then 
the proof is done.
Thus, we assume that 
${\sf def}(\mu) > 0$. 
Then we prove that 
there exists an element $\varphi \in \Phi$ such that 
${\sf def}(\mu) > {\sf def}(\varphi)$. 
This contradicts the definition of $\mu$.

For every edge $e = (d,h) \in \mu \setminus \sigma$,
since 
${\sf P}_{c,i_c}$ does not contain a bad edge,
$e \succsim_d \kappa(d) \succsim_d \sigma(d)$.
Furthermore, if 
$e \not \succ_d \sigma(d)$
(that is, $e \sim_d \sigma(d)$), 
then $\sigma(d) \in K$.

\begin{lemma} \label{lemma:dicycle} 
Let $e = (d,h)$ be an edge 
in $\mu \setminus \sigma$ such that 
$e \sim_d \sigma(d)$. 
Then $\sigma(d) \in {\sf sp}_{{\bf H}\langle K \rangle}(\mu)$, 
and there exists an edge $f = (s,p) \in \mu \setminus \sigma$
such that 
$f \in {\sf D}_{{\bf H}\langle K \rangle}(\sigma(d), \mu)$
and $f \sim_s \sigma(s)$.  
\end{lemma}
\begin{proof}
Recall that 
$\mu$ is a base of ${\bf H}\langle K \rangle$. 
Thus, $\sigma(d) \in {\sf sp}_{{\bf H}\langle K \rangle}(\mu)$.  
Define $C := {\sf C}_{{\bf H}\langle K \rangle}(\sigma(d),\mu)$.
For every edge $f \in C$, 
$f \sim_H \sigma(d)$. 

Define $C^{\prime} := {\sf C}_{\bf H}(\sigma(d),\mu)$. 
Lemma~\ref{lemma:new_matroid_circuit} implies 
that $C \subseteq C^{\prime}$
and
$f \succ_H \sigma(d)$
for every edge $f \in C^{\prime} \setminus C$. 
For every edge $f = (s,p) \in C^{\prime} \setminus (C \cup \sigma)$, 
since 
$f \succsim_s \sigma(s)$ and 
$\sigma$ is strongly stable,  
$f \in {\sf sp}_{\bf H}(\sigma)$ and 
$g \succsim_H f$
for every edge $g \in {\sf D}_{\bf H}(f,\sigma)$. 
Thus, 
for every edge $f \in C^{\prime} \setminus (C \cup \sigma)$, 
since $f \succ_H \sigma(d)$, 
$\sigma(d) \notin {\sf C}_{\bf H}(f,\sigma)$. 

Assume that, 
for every edge $f = (s,p) \in C \setminus \sigma$, 
$f \succ_{s} \sigma(s)$. 
Then for every edge $f \in C \setminus \sigma$, 
since $\sigma$ is strongly stable, 
$f \in {\sf sp}_{\bf H}(\sigma)$ and 
$g \succ_H f$
for every edge $g \in {\sf D}_{\bf H}(f,\sigma)$.
Thus, 
for every edge $f \in C \setminus \sigma$, 
since $f \sim_H \sigma(d)$, 
$\sigma(d) \notin {\sf C}_{\bf H}(f,\sigma)$. 

These imply that 
there exists a circuit $C^{\circ}$ of ${\bf H}$ 
such that 
$C^{\circ} \subseteq \sigma$.
This contradicts the fact that $\sigma \in \mathcal{F}$. 
\end{proof} 

Lemma~\ref{lemma:dicycle} implies that there exists a simple directed 
cycle $P$ in 
${\sf G}_{{\bf D}|K,{\bf H}\langle K \rangle}(\mu)$
satisfying the following conditions. 
Let $X$ be the set of edges in $K$ (that is, vertices 
of ${\sf G}_{{\bf D}|K,{\bf H}\langle K \rangle}(\mu)$)
that 
$P$ passes through.  
\begin{itemize}
\item 
$X \subseteq (\sigma \setminus \mu) \cup (\mu \setminus \sigma)$.
\item
For every edge $e = (d,h) \in X \cap \mu$,
$\sigma(d) \in X$ and 
$e \sim_d \sigma(d)$. 
\item
There does not exist a shortcut arc in 
${\sf A}_{{\bf D}|K,{\bf H}\langle K \rangle}(\mu)$ 
for $P$. 
\end{itemize} 
Define $\varphi := \mu \ominus P$. 
Then Lemma~\ref{lemma:IriT76_intersection}
implies that 
$\varphi$ is a common independent set 
of ${\bf D}|K$ and 
${\bf H}\langle K \rangle$
such that $|\varphi| = |\mu| = |\kappa|$,  
$\varphi(d) \sim_d \mu(d) \succsim_d \kappa(d)$ for every doctor $d \in D$, 
and ${\sf def}(\mu) > {\sf def}(\varphi)$.  
What remains is to prove that 
$\varphi$ satisfies (P1), (P2), and (P3). 

We assume that 
$X \cap \sigma = \{e_1,e_2,\dots,e_{\ell}\}$, 
$X \cap \mu = \{f_1,f_2,\dots,f_{\ell}\}$, 
$e_if_i \in {\sf A}_{{\bf D}|K,{\bf H}\langle K \rangle}(\mu)$
for every integer $i \in [\ell]$, and 
$P$ passes through $e_1,e_2,\dots,e_{\ell}$ in this order. 
For every integer $i \in [\ell]$, 
since we have $f_i \in {\sf D}_{{\bf H}\langle K \rangle}(e_i,\mu)$, 
Lemma~\ref{lemma:new_matroid_circuit} implies that 
$e_i \sim_H f_i$. 

For each integer $i \in [\ell]$, 
we define 
\begin{equation*}
\mu_i := (\mu \cup \{e_1,e_2,\dots,e_i\}) \setminus \{f_1,f_2,\dots,f_i\}.
\end{equation*}
Define $\mu_0 := \mu$. 
Since $\mu$ is a base of ${\bf H}\langle K \rangle$, 
Lemma~\ref{lemma:IriT76_sequence} implies that, 
for every integer $i \in [\ell]$, 
$\mu_i$ is a base of 
${\bf H}\langle K \rangle$. 
Thus, 
for every integer $i \in \{0\} \cup [\ell-1]$, 
$e_{i+1} \in {\sf sp}_{{\bf H}\langle K \rangle}(\mu_i)$ and 
$f_{i+1} \in {\sf C}_{{\bf H}\langle K \rangle}(e_{i+1},\mu_i)$. 

We prove that, for every integer $i \in [\ell]$, 
$\mu_i$ satisfies (P1), (P2), and (P3).  
Notice that $\mu_{\ell} = \varphi$. 
Then $\mu_0$ satisfies (P1), (P2), and (P3). 
Assume that, for some integer $z \in \{0\} \cup [\ell-1]$, 
$\mu_{z}$ satisfies (P1), (P2), and (P3).
We prove that $\mu_{z+1} = \mu_{z} + e_{z+1} - f_{z+1}$
satisfies (P1), (P2), and (P3). 
Let $\eta$ be an edge in 
$\sigma \cap {\sf tail}_H({\sf C}_{\bf H}(b,\mu_{z}))$. 
Then (P2) implies that 
$\xi \succsim_H \eta$. 

We prove that $\mu_{z+1}$ satisfies (P1). 
Since $f_{z+1} \in {\sf C}_{{\bf H}\langle K \rangle}(e_{z+1},\mu_z)$, 
Lemma~\ref{lemma:new_matroid_circuit} implies that 
$f_{z+1} \in {\sf C}_{{\bf H}}(e_{z+1},\mu_z)$.
Thus, 
Lemma~\ref{lemma:IriT76_sequence} implies that 
since $b \in {\sf sp}_{\bf H}(\mu_z)$, 
it is sufficient to prove that 
$b \neq e_{z+1}$. 
If $b = e_{z+1}$, then $b \in \sigma$.
However, Lemma~\ref{lemma_1:case_2}
implies that $b \notin \sigma$. 

We divide the remaining proof into the following 
three cases. 
\begin{equation*}
\mbox{(1)} \ \ f_{z+1} \succ_H \eta, \ \ \ \ 
\mbox{(2)} \ \ \eta \succ_H f_{z+1}, \ \ \ \ 
\mbox{(3)} \ \ f_{z+1} \sim_H \eta. 
\end{equation*}
Recall that $e_{z+1} \sim_H f_{z+1}$. 
In what follows, we define $C_z := {\sf C}_{\bf H}(b,\mu_z)$ 
and $C_{z+1} := {\sf C}_{\bf H}(b,\mu_{z+1})$. 

{\bf Case (1).}
We prove that
$\{g \in C_z \mid \eta \succsim_H g\}
= 
\{g \in C_{z+1} \mid \eta \succsim_H g\}$. 
If we can prove this, then 
${\sf tail}_H(C_z) = {\sf tail}_H(C_{z+1})$.
Thus, $\mu_{z+1}$ satisfies (P2) and (P3). 
To this end, Lemma~\ref{lemma:IriT76_auxiliary}(2)
implies that it is sufficient to prove that
$g \notin {\sf C}_{\bf H}(e_{z+1}, \mu_z)$ 
for every edge 
$g \in \mu_z$ 
such that $\eta \succsim_H g$. 
Let $g$ be an edge in $\mu_z$ such that $\eta \succsim_H g$.
Then since $e_{z+1} \in K \setminus \mu_{z}$, 
Lemma~\ref{lemma:new_matroid_circuit} implies that
$g^{\prime} \succsim_H e_{z+1}$
for every edge $g^{\prime} \in {\sf C}_{\bf H}(e_{z+1},\mu_z)$. 
Thus, since $e_{z+1} \succ_H \eta \succsim_H g$, 
$g \notin {\sf C}_{\bf H}(e_{z+1},\mu_z)$. 

{\bf Case (2).} 
We prove that 
$C_z = C_{z+1}$. 
Clearly, this implies that 
$\mu_{z+1}$ satisfies (P2) and (P3). 
To this end, 
Lemma~\ref{lemma:IriT76_auxiliary}
implies that it is sufficient to prove that
$f_{z+1} \notin C_z$. 
This is because 
Lemma~\ref{lemma:IriT76_auxiliary}(1) 
implies that 
$e_{z+1} \notin C_{z+1}$, and 
Lemma~\ref{lemma:IriT76_auxiliary}(2)
implies that, for every 
edge $g \in \mu_z - f_{z+1}$, 
$g \in C_z$ if and only if 
$g \in C_{z+1}$. 
If $f_{z+1} \in C_z$, then 
since $\eta \succ_H f_{z+1}$, 
$\eta \notin {\sf tail}_H(C_z)$.
However, this contradicts the fact that 
$\eta \in {\sf tail}_H(C_z)$.

{\bf Case (3).} 
If $f_{z+1} \notin C_z$, then 
$f_{z+1} \neq \eta$ and 
Lemma~\ref{lemma:IriT76_auxiliary}(2) implies that 
$\eta \in C_{z+1}$. 
On the other hand, 
if $f_{z+1} \in C_z$, then 
Lemma~\ref{lemma:IriT76_auxiliary}(1)
implies that 
$e_{z+1} \in C_{z+1}$. 
Notice that $\xi \succsim_H \eta \sim_H e_{z+1}$. 
Furthermore, for every edge $g \in \mu_z$ such that 
$\eta \succ_H g$, 
Lemma~\ref{lemma:new_matroid_circuit} implies that
$g \notin {\sf C}_{\bf H}(e_{z+1},\mu_z)$, and 
Lemma~\ref{lemma:IriT76_auxiliary}(2)
implies that $g \notin C_{z+1}$. 
Thus, 
in both cases, 
at least one of 
$\eta \in {\sf tail}_H(C_{z+1})$ and  
$e_{z+1} \in {\sf tail}_H(C_{z+1})$ holds.  
Since $\eta, e_{z+1} \in \sigma$, 
$\mu_{z+1}$ satisfies (P2) and (P3).
\end{proof} 

\section{Conclusion} 

In this paper, we propose a polynomial-time algorithm for
finding a strongly stable matching under matroid constraints. 
It would be interesting to consider an extension of the result in this paper 
to the many-to-many setting. 

\appendix

\section{The Student-Project Allocation Problem with Ties} 
\label{appendix:spa}

Here we define 
the {\em student-project allocation problem with ties}~\cite{OM20}
with our terminology. 
More concretely, we explain how to define 
$\succsim_H$ and ${\bf H}$ in this problem. 
In this problem, we are given 
a finite set $R = [m]$ of {\em regions}, and 
$H$ is  partitioned into  
$H_1, H_2, \ldots, H_m$. 
For each subset $F \subseteq E$ and 
each region $r \in R$, 
we define $F(r) := \bigcup_{h \in H_r}F(h)$. 
Then for each region $r \in R$, 
we define $D_r$ as the set of doctors $d \in D$ such that 
$E(d) \cap E(r) \neq \emptyset$. 
For each region $r \in R$, we are given a 
transitive binary relation $\succsim_{r}$ on $D_r$
such that, 
for every pair of doctors $d,s \in D_r$, 
at least one of $d \succsim_{r} s$ and 
$s \succsim_{r} d$ holds. 
We are given capacity functions 
$c_H \colon H \to \mathbb{Z}_+$ and 
$c_R \colon R \to \mathbb{Z}_+$.
Then $\mathcal{F}$ is defined as the set of 
subsets $F \subseteq E$ 
such that 
$|F(h)| \le c_H(h)$ for every hospital $h \in H$, and 
$|F(r)| \le c_R(r)$ for every region $r \in R$.
It is not difficult to see that 
${\bf H} = (E, \mathcal{F})$ is a matroid. 
Furthermore, $\succsim_H$ is defined as follows. 
For each region $r \in R$ and each pair of 
edges $(d,h), (s,p) \in E(r)$,
$(d,h) \succsim_H (s,p)$ if and only if 
$d \succsim_r s$. 
For each pair of distinct regions $r,r^{\prime} \in R$ such that 
$r < r^{\prime}$ and each pair of edges $e \in E(r)$
and $f \in E(r^{\prime})$,  
$e \succ_H f$. 

Let $\mu$ be a matching in $G$.
Then in this setting, an edge $(d,h) \in E \setminus \mu$
weakly (resp.\ strongly) blocks $\mu$ on ${\bf H}$ if
one of the following conditions is satisfied. 
Let $r$ be the region in $R$ such that 
$h \in H_r$. 
\begin{itemize}
\item[\bf (S1)]
$|\mu(h)| < c_H(h)$
and $|\mu(r)| < c_R(r)$. 
\item[\bf (S2)]
$|\mu(h)| < c_H(h)$, 
$|\mu(r)| = c_R(r)$, and 
there is an edge  
$(s,p) \in \mu(r)$ such that 
$d \succsim_r s$ 
(resp.\ $d \succ_r s$).
\item[\bf (S3)] 
$|\mu(h)| = c_H(h)$, and 
there is an edge $(s,h) \in \mu(h)$
such that 
$d \succsim_r s$ (resp.\ $d \succ_r s$).
\end{itemize}

There exists a difference between 
our definition of a blocking edge 
and the definition in \cite{OM20}.
In (S1) and (S2) of 
the definition in \cite{OM20}, 
when $(d,h)$ strongly blocks $\mu$ on ${\bf H}$, 
we are given the additional condition that 
$\mu(d) \notin E(r)$. Thus, 
our definition of 
a blocking edge and the definition in \cite{OM20} 
are slightly different. 
(For example, the strongly stable matching in 
\cite[Section~3.3]{OM20} is not strongly stable 
in our definition because 
$(s_4,p_6)$ is a blocking pair.)

\section{Example} 
\label{section:example} 

Here we give an example of execution of 
Algorithm~\ref{alg:main}.
We consider the following instance, which is obtained 
by slightly modifying the instance in  \cite[Section~3.3]{OM20}. 
Basically, the following instance is obtained from 
the instance in  \cite[Section~3.3]{OM20} by changing 
$c_H(h_6)$ from $2$ to $1$. 
\begin{equation*}
\begin{split}
& D = \{d_1,d_2,\dots,d_8\}, \  
H = \{h_1,h_2,\dots,h_6\}, \ 
R = \{1,2,3\}, \\
& H_1 = \{h_1,h_2\}, \ 
H_2 = \{h_3,h_4\}, \ 
H_3 = \{h_5,h_6\}. 
\end{split}
\end{equation*}
The capacities are defined as follows. 
\begin{equation*}
\begin{split}
& c_H(h_1) = c_H(h_2) = 2, \
c_H(h_3) = c_H(h_4) = c_H(h_5) = c_H(h_6) = 1, \\
& c_R(1) = 3, \
c_R(2) = c_R(3) = 2. 
\end{split}
\end{equation*}
The preferences of doctors are defined as follows.
(The subscripts and the first elements of 
edges are omitted in the definition.)
\begin{equation*}
\begin{split}
& \succsim_{d_1} \colon h_1 \succ h_6, \ \
\succsim_{d_2} \colon h_1 \succ h_2, \ \
\succsim_{d_3} \colon h_1 \sim h_4, \ \ 
\succsim_{d_4} \colon h_2 \succ h_5 \sim h_6, \\ 
& \succsim_{d_5} \colon h_2 \sim h_3, \ \
\succsim_{d_6} \colon h_2 \sim h_4, \ \
\succsim_{d_7} \colon h_3 \succ h_1, \ \ 
\succsim_{d_8} \colon h_5 \succ h_1. 
\end{split}
\end{equation*}
The preferences of regions 
are defined as follows.
\begin{equation*}
\begin{split}
& \succsim_{1} \colon 
d_8 \succ_{1} d_7 \succ_{1} d_1 \sim_{1} d_2 \sim_{1} d_3 
\succ_{1} d_4 \sim_{1} d_5 \succ_{1} d_6,\\
& \succsim_{2} \colon 
d_6 \succ_{2} d_5 \succ_{2} d_7 \sim_{2} d_3,\\
& \succsim_{3} \colon 
d_1 \sim_{3} d_4 \succ_{3} d_8.
\end{split}
\end{equation*}

When $t = 1$ and $i = 1$, we have 
\begin{equation*}
\begin{split}
K_{1,1} & = \{(d_1,h_1), (d_2,h_1), (d_3,h_1), (d_3,h_4), (d_4,h_2),(d_5,h_2), \\
& \ \ \ \ \ \ \ 
(d_5,h_3), (d_6,h_2), (d_6,h_4), (d_7,h_3), (d_8,h_5)\},\\
\kappa_{1,1} & = \{(d_1,h_1), (d_2,h_1), (d_4,h_2), (d_5,h_3), (d_6,h_4), (d_8,h_5)\},\\
Z_{1,1} & = \{d_1,d_2,d_3,d_7\}, \ \ 
\rho_{K_{1,1}}(Z_{1,1}) = 2 - 4 = -2,\\
{\sf P}_{1,1} & = \{(d_1,h_1), (d_2,h_1), (d_3,h_1), (d_3,h_4), (d_7,h_3)\}.  
\end{split}
\end{equation*}

When $t = 1$ and $i = 2$, we have 
\begin{equation*}
\begin{split}
K_{1,2} & = \{(d_1,h_6), (d_2,h_2), (d_4,h_2), (d_5,h_2), (d_5,h_3), (d_6,h_2),\\
& \ \ \ \ \ \ \ 
(d_6,h_4), (d_7,h_1), (d_8,h_5)\},\\
\kappa_{1,2} & = \{(d_1,h_6), (d_2,h_2), (d_4,h_2), (d_5,h_3), (d_6,h_4), (d_7,h_1),\\
& \ \ \ \ \ \ \ 
(d_8,h_5)\}.
\end{split}
\end{equation*}
Then we have 
$b_1 = (d_1,h_1)$ and ${\sf R}_1 \setminus {\sf P}_{1,2} = \{(d_4,h_2)\}$. 

When $t = 2$ and $i = 1$, we have 
\begin{equation*}
\begin{split}
& K_{2,1} = \{(d_1,h_6), (d_2,h_2), (d_4,h_5), (d_4,h_6), (d_5,h_2), (d_5,h_3),\\
& \ \ \ \ \ \ \ \ \ \ \ \ \ 
(d_6,h_2), (d_6,h_4), (d_7,h_1), (d_8,h_5)\},\\
& \kappa_{2,1} = \{(d_1,h_6), (d_2,h_2), (d_4,h_5), (d_5,h_2), (d_6,h_4), (d_7,h_1)\},\\
& Z_{2,1} = \{d_8\}, \ \ 
\rho_{K_{2,1}}(Z_{2,1}) = 0 - 1 = -1,\\
& {\sf P}_{2,1} \setminus {\sf R}_1 = \{(d_8,h_5)\}.  
\end{split}
\end{equation*}

When $t = 2$ and $i = 2$, we have 
\begin{equation*}
\begin{split}
K_{2,2} & = \{(d_1,h_6), (d_2,h_2), (d_4,h_5), (d_4,h_6), (d_5,h_2), (d_5,h_3),\\
& \ \ \ \ \ \ \ 
(d_6,h_2), (d_6,h_4), (d_7,h_1), (d_8,h_1)\},\\
\kappa_{2,2} & = \{(d_1,h_6), (d_2,h_2), (d_4,h_5), (d_5, h_3), (d_6, h_4), (d_7,h_1),\\
& \ \ \ \ \ \ \ 
(d_8,h_1)\}.
\end{split}
\end{equation*}
Then $\kappa_{2,2}$ is a strongly stable matching in $G$. 

On the other hand, if $c_H(h_6) = 2$ and 
$c_R(3) = 3$, then 
${\sf r}_{\bf H}(K_{2,1}) > |D[E \setminus {\sf P}_{2,0}]|$, and 
Algorithm~\ref{alg:main} outputs {\bf null}. 
For example, $(d_4,h_6)$ is a blocking edge for 
$\kappa_{2,2}$.


\begin{thebibliography}{10}

\bibitem{Bli21}
Joakim Blikstad.
\newblock {Breaking {$O(nr)$} for Matroid Intersection}.
\newblock In Nikhil Bansal, Emanuela Merelli, and James Worrell, editors, {\em
  Proceedings of the 48th International Colloquium on Automata, Languages, and
  Programming}, volume 198 of {\em Leibniz International Proceedings in
  Informatics}, pages 31:1--31:17, Wadern, Germany, 2021. Schloss Dagstuhl --
  Leibniz-Zentrum f{\"u}r Informatik.

\bibitem{ChenG10}
Ning Chen and Arpita Ghosh.
\newblock Strongly stable assignment.
\newblock In Mark de~Berg and Ulrich Meyer, editors, {\em Proceedings of the
  18th Annual European Symposium on Algorithms, Part {II}}, volume 6347 of {\em
  Lecture Notes in Computer Science}, pages 147--158, Berlin, Heidelberg,
  Germany, 2010. Springer.

\bibitem{C86}
William~H. Cunningham.
\newblock Improved bounds for matroid partition and intersection algorithms.
\newblock {\em SIAM Journal on Computing}, 15(4):948--957, 1986.

\bibitem{F03}
Tam{\'{a}}s Fleiner.
\newblock A fixed-point approach to stable matchings and some applications.
\newblock {\em Mathematics of Operations Research}, 28(1):103--126, 2003.

\bibitem{FK16}
Tam{\'{a}}s Fleiner and Naoyuki Kamiyama.
\newblock A matroid approach to stable matchings with lower quotas.
\newblock {\em Mathematics of Operations Research}, 41(2):734--744, 2016.

\bibitem{FT07}
Satoru Fujishige and Akihisa Tamura.
\newblock A two-sided discrete-concave market with possibly bounded side
  payments: An approach by discrete convex analysis.
\newblock {\em Mathematics of Operations Research}, 32(1):136--155, 2007.

\bibitem{GS62}
David Gale and Lloyd~S. Shapley.
\newblock College admissions and the stability of marriage.
\newblock {\em The American Mathematical Monthly}, 69(1):9--15, 1962.

\bibitem{H10}
Chien{-}Chung Huang.
\newblock Classified stable matching.
\newblock In Moses Charikar, editor, {\em Proceedings of the 21st Annual
  {ACM-SIAM} Symposium on Discrete Algorithms}, pages 1235--1253, Philadelphia,
  PA, 2010. Society for Industrial and Applied Mathematics.

\bibitem{IriT76}
Masao Iri and Nobuaki Tomizawa.
\newblock An algorithm for finding an optimal {``Independent Assignment''}.
\newblock {\em Journal of The Operations Research Society of Japan}, 19:32--57,
  1976.

\bibitem{I94}
Robert~W. Irving.
\newblock Stable marriage and indifference.
\newblock {\em Discrete Applied Mathematics}, 48(3):261--272, 1994.

\bibitem{IMS00}
Robert~W. Irving, David~F. Manlove, and Sandy Scott.
\newblock The hospitals/residents problem with ties.
\newblock In Magn{\'{u}}s~M. Halld{\'{o}}rsson, editor, {\em Proceedings of the
  7th Scandinavian Workshop on Algorithm Theory}, volume 1851 of {\em Lecture
  Notes in Computer Science}, pages 259--271, Berlin, Heidelberg, Germany,
  2000. Springer.

\bibitem{IMS03}
Robert~W. Irving, David~F. Manlove, and Sandy Scott.
\newblock Strong stability in the hospitals/residents problem.
\newblock In Helmut Alt and Michel Habib, editors, {\em Proceedings of the 20th
  Annual Symposium on Theoretical Aspects of Computer Science}, volume 2607 of
  {\em Lecture Notes in Computer Science}, pages 439--450, Berlin, Heidelberg,
  Germany, 2003. Springer.

\bibitem{IMS08}
Robert~W. Irving, David~F. Manlove, and Sandy Scott.
\newblock The stable marriage problem with master preference lists.
\newblock {\em Discrete Applied Mathematics}, 156(15):2959--2977, 2008.

\bibitem{IM08}
Kazuo Iwama and Shuichi Miyazaki.
\newblock Stable marriage with ties and incomplete lists.
\newblock In Ming{-}Yang Kao, editor, {\em Encyclopedia of Algorithms}.
  Springer, Boston, MA, 2008 edition, 2008.

\bibitem{IFF01}
Satoru Iwata, Lisa Fleischer, and Satoru Fujishige.
\newblock A combinatorial strongly polynomial algorithm for minimizing
  submodular functions.
\newblock {\em Journal of the {ACM}}, 48(4):761--777, 2001.

\bibitem{IwataY20}
Satoru Iwata and Yu~Yokoi.
\newblock Finding a stable allocation in polymatroid intersection.
\newblock {\em Mathematics of Operations Research}, 45(1):63--85, 2020.

\bibitem{K15}
Naoyuki Kamiyama.
\newblock Stable matchings with ties, master preference lists, and matroid
  constraints.
\newblock In Martin Hoefer, editor, {\em Proceedings of the 8th International
  Symposium on Algorithmic Game Theory}, volume 9347 of {\em Lecture Notes in
  Computer Science}, pages 3--14, Berlin, Heidelberg, Germany, 2015. Springer.

\bibitem{K19}
Naoyuki Kamiyama.
\newblock Many-to-many stable matchings with ties, master preference lists, and
  matroid constraints.
\newblock In Edith Elkind, Manuela Veloso, Noa Agmon, and Matthew~E. Taylor,
  editors, {\em Proceedings of the 18th International Conference on Autonomous
  Agents and Multiagent Systems}, pages 583--591, Richland, SC, 2019.
  International Foundation for Autonomous Agents and Multiagent Systems.

\bibitem{K20}
Naoyuki Kamiyama.
\newblock On stable matchings with pairwise preferences and matroid
  constraints.
\newblock In Amal El~Fallah Seghrouchni, Gita Sukthankar, Bo~An, and Neil
  Yorke{-}Smith, editors, {\em Proceedings of the 19th International Conference
  on Autonomous Agents and Multiagent Systems}, pages 584--592, Richland, SC,
  2020. International Foundation for Autonomous Agents and Multiagent Systems.

\bibitem{K21}
Naoyuki Kamiyama.
\newblock Envy-free matchings with one-sided preferences and matroid
  constraints.
\newblock {\em Operations Research Letters}, 49(5):790--794, 2021.

\bibitem{K22}
Naoyuki Kamiyama.
\newblock A matroid generalization of the super-stable matching problem.
\newblock {\em {SIAM} Journal on Discrete Mathematics}, 36(2):1467--1482, 2022.

\bibitem{KMMP07}
Telikepalli Kavitha, Kurt Mehlhorn, Dimitrios Michail, and Katarzyna~E. Paluch.
\newblock Strongly stable matchings in time {$O(nm)$} and extension to the
  hospitals-residents problem.
\newblock {\em ACM Transactions on Algorithms}, 3(2):Article 15, 2007.

\bibitem{KTY18}
Fuhito Kojima, Akihisa Tamura, and Makoto Yokoo.
\newblock Designing matching mechanisms under constraints: An approach from
  discrete convex analysis.
\newblock {\em Journal of Economic Theory}, 176:803--833, 2018.

\bibitem{Kunysz18}
Adam Kunysz.
\newblock An algorithm for the maximum weight strongly stable matching problem.
\newblock In Wen{-}Lian Hsu, Der{-}Tsai Lee, and Chung{-}Shou Liao, editors,
  {\em Proceedings of the 29th International Symposium on Algorithms and
  Computation}, volume 123 of {\em Leibniz International Proceedings in
  Informatics}, pages 42:1--42:13, Wadern, Germany, 2018. Schloss Dagstuhl -
  Leibniz-Zentrum f{\"{u}}r Informatik.

\bibitem{Kunysz19}
Adam Kunysz.
\newblock A faster algorithm for the strongly stable {$b$}-matching problem.
\newblock In Pinar Heggernes, editor, {\em Proceedings of the 11th
  International Conference on Algorithms and Complexity}, volume 11485 of {\em
  Lecture Notes in Computer Science}, pages 299--310, Cham, Switzerland, 2019.
  Springer.

\bibitem{KunyszPG16}
Adam Kunysz, Katarzyna~E. Paluch, and Pratik Ghosal.
\newblock Characterisation of strongly stable matchings.
\newblock In Robert Krauthgamer, editor, {\em Proceedings of the 27th Annual
  {ACM-SIAM} Symposium on Discrete Algorithms}, pages 107--119, Philadelphia,
  PA, 2016. Society for Industrial and Applied Mathematics.

\bibitem{Malhotra04}
Varun~S. Malhotra.
\newblock On the stability of multiple partner stable marriages with ties.
\newblock In Susanne Albers and Tomasz Radzik, editors, {\em Proceedings of the
  12th Annual European Symposium on Algorithms}, volume 3221 of {\em Lecture
  Notes in Computer Science}, pages 508--519, Berlin, Heidelberg, Germany,
  2004. Springer.

\bibitem{M99}
David~F. Manlove.
\newblock Stable marriage with ties and unacceptable partners.
\newblock Technical Report TR-1999-29, The University of Glasgow, Department of
  Computing Science, 1999.

\bibitem{M13}
David~F. Manlove.
\newblock {\em Algorithmics of Matching under Preferences}.
\newblock World Scientific, Singapore, 2013.

\bibitem{M03}
Kazuo Murota.
\newblock {\em Discrete Convex Analysis}, volume~10 of {\em SIAM Monographs on
  Discrete Mathematics and Applications}.
\newblock Society for Industrial and Applied Mathematics, Philadelphia, PA,
  2003.

\bibitem{MY15}
Kazuo Murota and Yu~Yokoi.
\newblock On the lattice structure of stable allocations in a two-sided
  discrete-concave market.
\newblock {\em Mathematics of Operations Research}, 40(2):460--473, 2015.

\bibitem{OM20}
Sofiat Olaosebikan and David~F. Manlove.
\newblock An algorithm for strong stability in the student-project allocation
  problem with ties.
\newblock In Manoj Changat and Sandip Das, editors, {\em Proceedings of the 8th
  Annual International Conference on Algorithms and Discrete Applied
  Mathematics}, volume 12016 of {\em Lecture Notes in Computer Science}, pages
  384--399, Cham, Switzerland, 2020. Springer.

\bibitem{OlaosebikanM22}
Sofiat Olaosebikan and David~F. Manlove.
\newblock Super-stability in the student-project allocation problem with ties.
\newblock {\em Journal of Combinatorial Optimization}, 43(5):1203--1239, 2022.

\bibitem{O07}
Gregg O'Malley.
\newblock {\em Algorithmic Aspects of Stable Matching Problems}.
\newblock PhD thesis, The University of Glasgow, 2007.

\bibitem{O11}
James~G. Oxley.
\newblock {\em Matroid Theory}.
\newblock Oxford University Press, Oxford, UK, 2nd edition, 2011.

\bibitem{R42}
Richard Rado.
\newblock A theorem on independence relations.
\newblock {\em The Quarterly Journal of Mathematics}, 13(1):83--89, 1942.

\bibitem{S00}
Alexander Schrijver.
\newblock A combinatorial algorithm minimizing submodular functions in strongly
  polynomial time.
\newblock {\em Journal of Combinatorial Theory, Series {B}}, 80(2):346--355,
  2000.

\bibitem{S05}
Sandy Scott.
\newblock {\em A Study of Stable Marriage Problems with Ties}.
\newblock PhD thesis, The University of Glasgow, 2005.

\bibitem{Y17}
Yu~Yokoi.
\newblock A generalized polymatroid approach to stable matchings with lower
  quotas.
\newblock {\em Mathematics of Operations Research}, 42(1):238--255, 2017.

\end{thebibliography}
\end{document}